\documentclass{article}
\usepackage{amsfonts}
\usepackage{amsmath}
\usepackage{graphicx}
\usepackage{float}

\newtheorem{theorem}{Theorem}
\newtheorem{acknowledgement}[theorem]{Acknowledgement}

\newtheorem{example}[theorem]{Example}

\newtheorem{proposition}[theorem]{Proposition}
\newtheorem{remark}[theorem]{Remark}

\newenvironment{proof}[1][Proof]{\noindent\textbf{#1.} }{\ \rule{0.5em}{0.5em}}
\begin{document}

\title{Source reconstruction using a bilevel optimisation method with a
smooth weighted distance function}
\author{ \and Niklas Br\"{a}nnstr\"{o}m$^{1}$ and Leif \AA\ Persson$^{1,2}$ 
\\
\\
$^{1}$Swedish Defence Research Agency, FOI,\\
SE-901 82 Ume\aa , Sweden\\
\\
$^{2}$Ume\aa\ University,\\
Department of Mathematics\\
and Mathematical Statistics,\\
SE-901 87 Ume\aa , Sweden}
\maketitle
\date{}

\begin{abstract}
We consider a bilevel optimatisation method for inverse linear atmospheric
dispersion problems where both linear and non-linear model parameters are to
be determined. We propose that a smooth weighted Mahalanobis distance
function is used and derive sufficient conditions for when the follower
problem has local strict convexity. A few toy-models are presented where
local strict convexity and ill-posedness of the inverse problem are
explored, indeed the smooth distance function is compared and contrasted to
linear and piecewise linear ones. The bilevel optimisation method is then
applied to sensor data collected in wind tunnel experiments of a neutral gas
release in urban environments (MODITIC).
\end{abstract}

\section{Introduction}

An inverse atmospheric dispersion problem is stated as follows: given the
topography, the meteorological conditions and a set of detector readings
determine where and when the hazardous substance was released and in which
quantities. The problem is known as a source reconstruction problem, and it
is easy to state but harder to solve due to being, like most inverse
problems, ill-posed \cite{Tikhonov1963},\cite{Enting2002}. Spurred by
various applications including the locating of industrial plants \cite%
{Marchuk1986}, determining the amount of radioactive nuclides released from
Chernobyl \cite{GHL} and Fukushima \cite{StohlEtAl}, pin-pointing nuclear
tests \cite{RingbomEtAl}, estimation of material released from volcanoes 
\cite{TheysEtAl},\cite{StohlEtAl2010},\cite{GrahnEtAl} a number of different
methods for addressing the inverse problem have been suggested. Even though
the main difference perhaps lies in the interpretation of the results these
methods are usually divided into two main categories: the probabilistic
approach with its Bayesian methods and the deterministic approach with its
optimisation methods. In the Bayesian setting a likelihood function is
calculated and weighted with any a priori information that one has at hand
to yield a posterior probability density function, which is then sampled to
yield an estimate of the sought source term (see e.g. \cite{Stuart2010} for
an introduction to general Bayesian inverse problems). In the deterministic
setting with optimisation methods a norm is devised under which the sensor
response of candidate sources is compared with the given sensor readings.
The candidate source that best fits the given sensor readings (minimizes the
distances under the chosen norm) is then deemed the solution to the inverse
problem.

For linear inverse atmospheric dispersion problems these methods come in
many different flavours and have often been devised with a given application
in mind: usually there are - a priori imposed - restrictions on the source
characteristics, e.g. the method may assume that the source is well
localised (located at a single point in space) and that the release was
instantaneous. For example Yee and coauthors have written a series of papers
adapting the Bayesian method to inverse dispersion problems of increasing
complexity \cite{KYL2007},\cite{Yee2007},\cite{YF2010},\cite{Yee2012} and 
\cite{Yee2012B}.

To use an optimisation method the inverse dispersion problem has to be cast
in a manner where the distance (under a chosen norm) function between model
sensor data and the given sensor readings can be minimized. Usually a least
squares solution is sought, and in \cite{BP2015} conditions under which the
least square problem is well defined is presented. Much of the literature
focuses on the problem where it is a-priori assumed that there is only a
single source, see e.g. \cite{RL1998},\cite{THG2007},\cite{AYH2007} and \cite%
{ISS2012}. There are however exceptions, for example in \cite{SSI2012} the
renormalisation method (least square method under the renormalisation norm)
presented in \cite{ISS2012} is generalised to cover an unknown number of
point sources, and in \cite{Bocquet2005} the space-time has been discretised
and an optimal source term is constructed by forming a union of (space-time)
grid sized point sources.

In this paper we make a contribution to the literature on optimisation
methods by applying a bilevel optimisation method, see \cite{Bard1998}, to a
linear inverse dispersion problem. A bilevel optimisation method splits the
optimisation problem in two: into a leader (upper level) problem and a
follower (lower level) problem and they are solved concurrently rather than
simultaneously. We consider dispersion problems where the source is a single
point source emitting at a constant rate. This problem is well suited to a
bilevel optimisation method where the follower problem concerns solving for
the emission rate and the leader problem pinpointing the location of the
source. For the bilevel optimisation method to work properly the follower
problem is required to have minima, ideally a strict minimum. We therefore
study local strict convexity of the follower problem, see Theorem \ref%
{thm:local_convexity} for sufficient conditions. We then explore the concept
of local strict convexity and its connection to ill-possedness of inverse
problems through a few toy-model examples. Following this the paper is
rounded off with the bilevel optimisation method being applied to dispersion
data from a series of wind tunnel experiments of urban environments of
varying complexity. The wind tunnel data was collected as part of the
European Defence Agency category B project MODITIC. In all cases the
boundary layer is neutrally stable, and we only consider cases where the
released gas is neutrally buoyant making the dispersion problem linear.

\section{Bilevel optimization problems}

A bilevel optimization problem is a constrained optimization problem where
the constraints also includes an optimization problem. The problem is
divided into an upper level or leader problem, with decision variables $%
\boldsymbol{x}\in X\subseteq \mathbb{R}^{n}$, and a lower level or follower
problem, with decision variables $\boldsymbol{y}\in Y\subseteq \mathbb{R}%
^{m} $. Here $X$ and $Y$ may be restricted to integers or nonnegative
values. We follow the notation of \cite{Bard1998}, p. 6. The leader problem
has the form 
\begin{eqnarray*}
V &=&\min_{\boldsymbol{x}\in X}F\left( \boldsymbol{x},\boldsymbol{y}\left( 
\boldsymbol{x}\right) \right) \\
\boldsymbol{G}\left( \boldsymbol{x},\boldsymbol{y}\left( \boldsymbol{x}%
\right) \right) &\leq &\boldsymbol{0}
\end{eqnarray*}%
where $F$ and $\boldsymbol{G}$, respectively, are the leader objective
function and constraint function, and $\boldsymbol{y}\left( \boldsymbol{x}%
\right) $ is an optimal solution to the follower problem%
\begin{eqnarray*}
v\left( \boldsymbol{x}\right) &=&\min_{\boldsymbol{y}\in Y}f\left( 
\boldsymbol{x},\boldsymbol{y}\right) \\
\boldsymbol{g}\left( \boldsymbol{x},\boldsymbol{y}\right) &\leq &\boldsymbol{%
0}\text{.}
\end{eqnarray*}%
An ambiguity occurs if the follower problem has several optimal solutions,
i.e., $\boldsymbol{y}\left( x\right) $ is set--valued. Then the follower is
indifferent towards these points, but the leader objective may be different
for different points in $\boldsymbol{y}\left( x\right) $, and there is no
way for the leader to direct the follower to the upper level optimal point.
Therefore, there may be no optimal solution to the bilevel program although
all functions are continuous and $X,Y$ are compact, cf. \cite{Bard1998}, p.
11. In our problem the sets $X,Y$ will be positive orthants, the follower
objective function $f\left( \boldsymbol{x},\boldsymbol{y}\right) $ will be a
Mahalanobis distance function measuring the discrepancy between model data
and measurements. The leader objective function $F$ will have the form%
\begin{equation*}
F\left( \boldsymbol{x},\boldsymbol{y}\right) =\exp \left( -\lambda \sum
y_{i}\right) +f\left( \boldsymbol{x},\boldsymbol{y}\right)
\end{equation*}%
hence minimizing the least square function, but penalizing large values of $%
\boldsymbol{y}$, which will act as a regularization of the problem ($\lambda
>0$ is a regularization parameter).

\section{The follower problem}

In the setting we are considering we have a priori assumed that the source
is a point source releasing a neutrally buoyant substance at a constant
rate. Under these assumptions the source location $\boldsymbol{x}\in \mathbb{%
R}^{d}$ is a\ nonlinear model parameter while the emission rate $\boldsymbol{%
y}\in \mathbb{R}^{1}$ is a linear model parameter. In general, our model
formulation allows a linear combination of basic sources, with a linear
positive weight vector $y\in \mathbb{R}^{n}$. To solve the inverse problem
we need a source-sensor relationship. Since the problem is linear the
source-sensor relationship is given a matrix relationship, which for
computational efficiency \cite{Marchuk1986} is expressed through the adjoint
formulation of the problem, thus $A:\mathbb{R}^{d}\rightarrow \mathbb{R}%
_{+}^{m\times n}$ is a matrix function with nonnegative elements (no sinks
are considered) and the adjoint model data $\boldsymbol{\mu }=\boldsymbol{%
\mu }\left( \boldsymbol{x},\boldsymbol{y}\right) =A\left( \boldsymbol{x}%
\right) \boldsymbol{y}\in \mathbb{R}^{m}$. The measured data $\boldsymbol{z}%
\in \mathbb{R}^{m}$ is the sensor response.\newline
We regard $\boldsymbol{z}$ as a random vector, and we assume that the
adjoint model data $\boldsymbol{\mu }$ represent the mean of $\boldsymbol{z}$%
. We also assume that the components $z_{i}$ of $\boldsymbol{z}$ are
statistically independent and that the variance of $z_{i}$ is%
\begin{equation*}
var\left( z_{i}\right) =\sigma ^{2}\left( \mu _{i}\right)
\end{equation*}%
where $\sigma :\mathbb{R\rightarrow R}_{+}$ is a given function. We let the
follower objective function be the \emph{Mahalanobis distance} between $%
\boldsymbol{z}$ and $\boldsymbol{\mu }$, viz., 
\begin{equation}
f\left( \boldsymbol{x},\boldsymbol{y}\right) =\sum_{i=1}^{m}\frac{\left(
z_{i}-\mu _{i}\right) ^{2}}{\sigma ^{2}\left( \mu _{i}\right) }\text{.}
\label{eqn:f}
\end{equation}%
In particular, we want to be able to choose a scale invariant distance
function, giving equal emphasis to all $\mu _{i}$, regardless of their size.

\subsection{Local convexity of the follower problem}

If the follower problem is strictly convex, the follower problem has a
unique optimal solution $\boldsymbol{y}\left( \boldsymbol{x}\right) $ for
each $\boldsymbol{x}\in X$, and by mild assumptions on $f\left( \boldsymbol{x%
},\boldsymbol{y}\right) $, an envelope theorem holds (e.g., \cite%
{MilgromSegal2002}, Theorem 2, p. 586), which implies that the optimal value
function%
\begin{equation*}
V\left( \boldsymbol{x}\right) =f\left( \boldsymbol{x},\boldsymbol{y}\left( 
\boldsymbol{x}\right) \right) =\inf_{\boldsymbol{y\in Y}}\;f\left( 
\boldsymbol{x},\boldsymbol{y}\right)
\end{equation*}%
is continuous. Therefore, the leader problem 
\begin{equation*}
\inf_{\boldsymbol{x\in X}}F\left( \boldsymbol{x},\boldsymbol{y}\left( 
\boldsymbol{x}\right) \right)
\end{equation*}%
has a solution since%
\begin{equation*}
F\left( \boldsymbol{x},\boldsymbol{y}\left( \boldsymbol{x}\right) \right)
=\exp \left( -\lambda \sum_{i}y_{i}\left( \boldsymbol{x}\right) \right)
V\left( \boldsymbol{x}\right) \text{. }
\end{equation*}%
However, in our setting the only situation when the follower problem is
guaranteed to be strictly convex is when $\sigma \left( \mu \right) $ is
constant, i.e., the classical least square method. However, we may derive
conditions for \emph{local convexity}, as the following theorem shows.

\begin{theorem}
\label{thm:local_convexity}Assume that 
\begin{equation}
f\left( y_{1},...,y_{n}\right) =\sum_{i=1}^{m}\frac{r_{i}^{2}}{\sigma
^{2}\left( \mu _{i}\right) }  \label{def:f}
\end{equation}%
where%
\begin{eqnarray*}
\mu _{i} &=&\sum_{j}a_{ij}y_{j} \\
r_{i} &=&z_{i}-\mu _{i}\text{. }
\end{eqnarray*}%
and $A=\left( a_{ij}\right) \in \mathbb{R}^{m\times n}$. Consider a fixed $%
\boldsymbol{y}$ and suppose that $\sigma $ is positive and twice
continuously differentiable in a neigbourhood of $\mu $.Then%
\begin{eqnarray*}
\frac{\partial f}{\partial y_{k}} &=&-2\sum_{i}a_{ik}\rho \left( \mu
_{i},r_{i}\right) \frac{r_{i}^{2}}{\sigma ^{2}\left( \mu _{i}\right) } \\
\frac{\partial ^{2}f}{\partial y_{k}\partial y_{l}} &=&2\sum_{i}a_{ik}a_{il}%
\eta \left( \mu _{i},r_{i}\right) \frac{r_{i}^{2}}{\sigma ^{2}\left( \mu
_{i}\right) }
\end{eqnarray*}%
where%
\begin{eqnarray*}
\rho \left( \mu ,r\right) &=&\frac{1}{r}+\frac{\sigma ^{\prime }\left( \mu
\right) }{\sigma \left( \mu \right) } \\
\eta \left( \mu ,r\right) &=&\left( \frac{1}{r}+2\frac{\sigma ^{\prime
}\left( \mu \right) }{\sigma \left( \mu \right) }\right) ^{2}-\left( \left( 
\frac{\sigma ^{\prime }\left( \mu \right) }{\sigma \left( \mu \right) }%
\right) ^{2}+\frac{\sigma ^{\prime \prime }\left( \mu \right) }{\sigma
\left( \mu \right) }\right) \text{.}
\end{eqnarray*}%
Moreover, if $A$ has full rank then

\begin{enumerate}
\item If $\eta \left( \mu _{i},r_{i}\right) >0$ for all $i$ then $f$ is
strictly convex in a neighborhood of $\boldsymbol{y}$.

\item If $\eta \left( \mu _{i},r_{i}\right) <0$ for all $i$ then $f$ is
strictly concave in a neigborhood of $\boldsymbol{y}$.
\end{enumerate}
\end{theorem}

\begin{proof}
The formulas for the first and second derivatives of $f$ are proved by
elementary but tedious calculations, using the chain rule and the quotient
rule for differentiation. If $A$ has full rank, and $\eta \left( \mu
_{i},r_{i}\right) >0$ for all $i$, then%
\begin{equation*}
H\left( \boldsymbol{u}\right) \equiv \sum_{k,l}u_{k}\frac{\partial ^{2}f}{%
\partial y_{k}\partial y_{l}}u_{l}=\sum_{i}\sum_{k,l}\psi
_{i}a_{ik}u_{k}\psi _{i}a_{il}u_{l}=\sum_{i}\left( \sum_{k}\psi
_{i}a_{ik}u_{k}\right) ^{2}
\end{equation*}%
for all $\boldsymbol{u}\in \mathbb{R}^{n}$,where%
\begin{equation*}
\psi _{i}=\sqrt{\eta \left( \mu _{i},r_{i}\right) \frac{r_{i}^{2}}{\sigma
^{2}\left( \mu _{i}\right) }}\text{.}
\end{equation*}%
Hence $H\left( \boldsymbol{u}\right) \geq 0$. Assume that $H\left( 
\boldsymbol{u}\right) =0$. Then $\sum_{k}\psi _{i}a_{ik}u_{k}=0$ for all $i$%
. Note that $\psi _{i}>0$, since when $r\rightarrow 0$ we have 
\begin{equation*}
\eta \left( \mu ,r\right) \frac{r^{2}}{\sigma ^{2}\left( \mu \right) }%
\rightarrow \frac{1}{\sigma ^{2}\left( \mu \right) }>0\text{. }
\end{equation*}%
Hence $\boldsymbol{u}=0$, since the vectors $\boldsymbol{v}_{i}=\left( \psi
_{i}a_{i1},\psi _{i}a_{i2},...,\psi _{i}a_{in}\right) $, $i=1,2,...,m$ span $%
\mathbb{R}^{n}$. This shows that $H$ is positive definite at $\boldsymbol{y}$%
. By continuity, $H$ is positive definite, and hence $f$ strictly convex, in
a neighborhood of $\boldsymbol{y}$. Similarly, if $A$ has full rank and $%
\eta \left( \mu _{i},r_{i}\right) <0$, then 
\begin{equation*}
H\left( \boldsymbol{u}\right) \equiv \sum_{k,l}u_{k}\frac{\partial ^{2}f}{%
\partial y_{k}\partial y_{l}}u_{l}=-\sum_{i}\left( \sum_{k}\psi
_{i}a_{ik}u_{k}\right) ^{2}
\end{equation*}%
where 
\begin{equation*}
\psi _{i}=\sqrt{-\eta \left( \mu _{i},r_{i}\right) \frac{r_{i}^{2}}{\sigma
^{2}\left( \mu _{i}\right) }}
\end{equation*}%
and it follows with a similar argument as above that $f$ is strictly concave
in a neighborhood of $\boldsymbol{y}$.
\end{proof}

Most inverse modelling methods works perfectly for synthetic data, i.e. when
the "observed" sensor response is calculated using a dispersion model (the
same dispersion model that is then used to solve the inverse problem).
Indeed showing that an inverse modelling method works well for synthetic
data and, in particular, slightly perturbed synthetic data is usually
included in the body of work motivating the method. In view of the theorem
we make the following observation.

\begin{remark}
Note that when $r\rightarrow 0$, then $\eta \left( \mu ,r\right) \rightarrow
\infty $, so for small enough $r$, $\eta $ is always positive, if $\mu $ is
restricted to a compact set. Hence for $\boldsymbol{z}$ sufficiently close
to $\mathbf{\mu }\left( \boldsymbol{y}\right) $, $f\left( \boldsymbol{y}%
\right) $ is convex. This explains why most methods work well on synthetic
model data with small perturbations.
\end{remark}

We will now restrict ourselves to functions $\sigma \left( \mu \right) $
which are continuous, increasing, convex, and positive for $\mu >0$, and
satisfying%
\begin{equation*}
\lim_{\mu \rightarrow \infty }\frac{\sigma \left( \mu \right) }{\mu }=1\text{%
. }
\end{equation*}%
Moreover, we assume that $\sigma $ is twice continuously differentiable,
except possibly at a finite number of points. Hence $\sigma ^{\prime \prime
}\left( \mu \right) \geq 0$ at all points where $\sigma ^{\prime \prime }$
exists. Note that if $\sigma \left( \mu \right) $ satisfies these
conditions, so does 
\begin{equation*}
\sigma _{\delta }\left( \mu \right) \equiv \delta \sigma \left( \mu /\delta
\right)
\end{equation*}%
for any $\delta >0$, so the class of functions $\sigma $ we consider is
scale--invariant. We consider three examples:%
\begin{eqnarray}
\sigma \left( \mu \right) &=&erf\left( \mu \right) -2\left( 1-\frac{1%
}{\sqrt{\pi }}\right) \exp \left( -\mu ^{2}\right) \text{ (Smooth threshold)}
\notag \\
\sigma \left( \mu \right) &=&\max \left( 1,\mu \right) \text{ (Piecewise
linear)}  \label{eqn:thesigmas} \\
\sigma \left( \mu \right) &=&\mu \text{ (Linear)}  \notag
\end{eqnarray}%

\begin{figure}[H]
\centering
\includegraphics[scale=0.25]{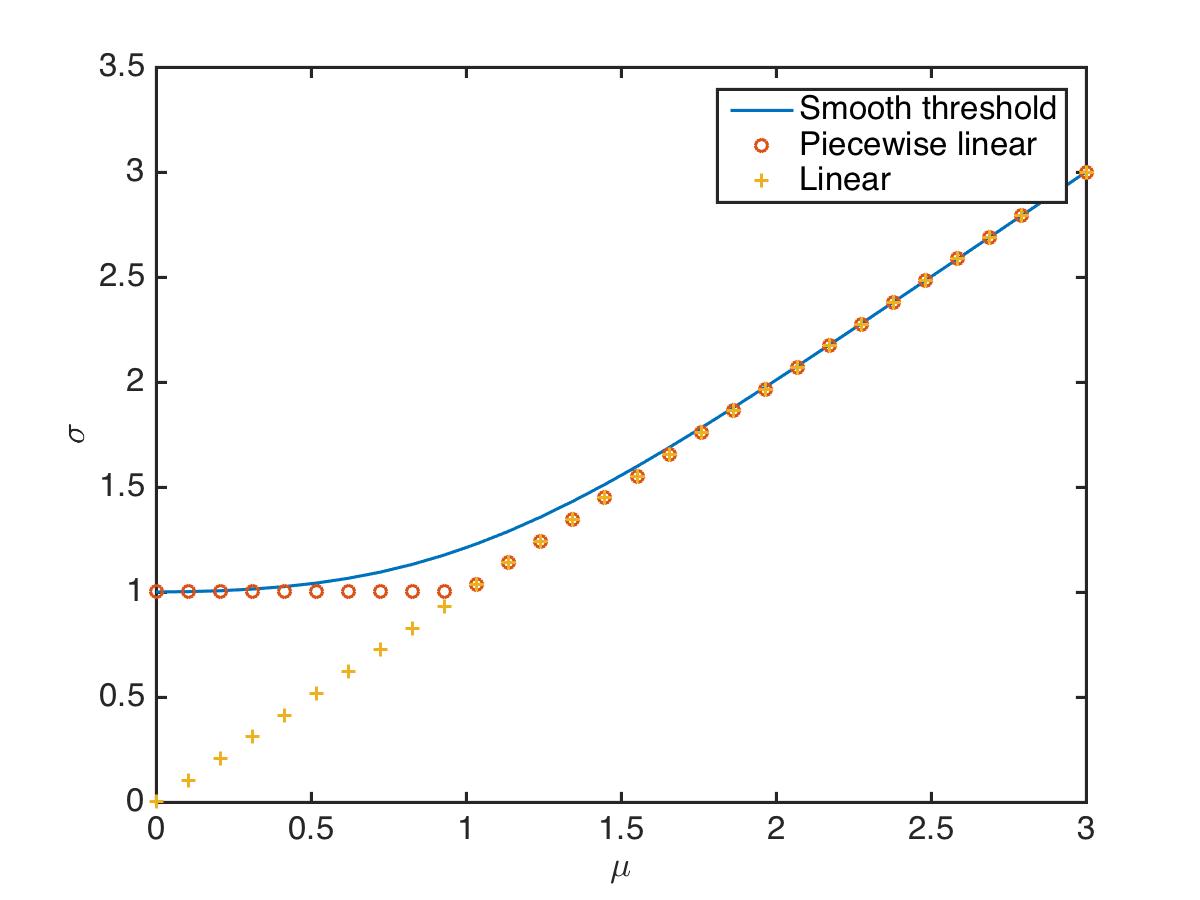}
\caption{Graphs of the three considered $%
\protect\sigma $--functions.}
\label{fig_simplex}
\end{figure}

Recall that

\begin{equation*}
f\left( \boldsymbol{y}\right) =\sum_{i}\frac{\left( z_{i}-\mu _{i}\right)
^{2}}{\sigma ^{2}\left( \mu _{i}\right) }
\end{equation*}%
Since we assume that $\sigma ^{\prime \prime }\geq 0$ and $\sigma >0$ for $%
\mu >0$, we can write $\eta $ on the form%
\begin{equation*}
\eta =\left( \frac{1}{r}-\alpha \right) \left( \frac{1}{r}-\beta \right)
\end{equation*}%
where%
\begin{eqnarray*}
\alpha \left( \mu \right) &=&-2\frac{\sigma ^{\prime }\left( \mu \right) }{%
\sigma \left( \mu \right) }-\sqrt{\left( \frac{\sigma ^{\prime }\left( \mu
\right) }{\sigma \left( \mu \right) }\right) ^{2}+\frac{\sigma ^{\prime
\prime }\left( \mu \right) }{\sigma \left( \mu \right) }} \\
\beta \left( \mu \right) &=&-2\frac{\sigma ^{\prime }\left( \mu \right) }{%
\sigma \left( \mu \right) }+\sqrt{\left( \frac{\sigma ^{\prime }\left( \mu
\right) }{\sigma \left( \mu \right) }\right) ^{2}+\frac{\sigma ^{\prime
\prime }\left( \mu \right) }{\sigma \left( \mu \right) }}
\end{eqnarray*}%
Clearly, $\eta >0$ if and only if the two factors have the same sign, i.e., $%
1/r\in \left[ \alpha ,\beta \right] ^{c}$, and $\eta <0$ if and only if $%
1/r\in \left( \alpha ,\beta \right) $. To analyze these conditions further
we need to consider the cases when $\alpha ,\beta $ have the same sign and
when they have different signs. If $\alpha ,\beta $ have the same sign then $%
1/r\in \left( \alpha ,\beta \right) $ if and only if $r\in \left( 1/\beta
,1/\alpha \right) $ (this includes the limiting cases $\beta \nearrow 0$ or $%
\alpha \searrow 0$, with $1/\beta =-\infty $ and $1/\alpha =\infty $). If $%
\alpha ,\beta $ have different signs, i.e., $\alpha <0<\beta $, then $1/r\in
\left( \alpha ,\beta \right) $ if and only if $r\in \left[ 1/\alpha ,1/\beta %
\right] ^{c}$. By the assumptions on $\sigma $, $\alpha $ and $\beta $ are
defined everywhere except possibly at a finite number of points. Henceforth,
we only consider points where $\alpha $ and $\beta $ are defined. We can now
conclude that $\eta <0$ if and only if either i) $z-\mu \in \left( \beta
^{-1},\alpha ^{-1}\right) $ and $\alpha <\beta \leq 0$, or ii) $z-\mu \in %
\left[ \alpha ^{-1},\beta ^{-1}\right] ^{c}$ and $\alpha <0<\beta $, or iii) 
$z-\mu \in \left( \beta ^{-1},\infty \right) $ and $0=\alpha <\beta $.
Likewise, we can conclude by complementarity that $\eta >0$ if and only if
either i) $z-\mu \in \left[ \beta ^{-1},\alpha ^{-1}\right] ^{c}$ and $%
\alpha <\beta \leq 0$, or ii) $z-\mu \in \left( \alpha ^{-1},\beta
^{-1}\right) $ and $\alpha <0<\beta $, or iii) $z-\mu \in \left( -\infty
,\beta ^{-1}\right) $ and $0=\alpha <\beta $ or iv) $\alpha =\beta =0$.

Relying on these inequalities we can find domains where a distance function
is guaranteed to be strictly convex or strictly concave.

\begin{example}
Consider $\sigma \left( \mu \right) =\mu $, the linear case. Then $\alpha
(\mu )=-3/\mu $ and $\beta (\mu )=-1/\mu $. Hence $\alpha ^{2}>\beta $ and $%
\eta \left( \mu ,r\right) >0$ if and only if $r\in \left( -\infty ,-\mu
\right) \cup \left( -\mu /3,\infty \right) $, i.e., $z=\mu +r\in \left(
-\infty ,0\right) \cup \left( 2\mu /3,\infty \right) $. Moreover, $\eta
\left( \mu ,r\right) <0$ if and only if $z=\mu +r\in \left( 0,2\mu /3\right) 
$. Hence for fixed $\boldsymbol{z}$, $f$ is strictly convex on the set $\cap
_{i}\left\{ 2\mu _{i}\left( \boldsymbol{y}\right) /3<z_{i}\right\} $, and
strictly concave on the set $\cap _{i}\left\{ 2\mu _{i}\left( \boldsymbol{y}%
\right) /3>z_{i}\right\} $. Let for example%
\begin{equation*}
A=%
\begin{bmatrix}
1 & 0 \\ 
0 & 1 \\ 
1 & 1%
\end{bmatrix}%
\text{, }z=%
\begin{bmatrix}
1 \\ 
1 \\ 
5/3%
\end{bmatrix}%
\text{.}
\end{equation*}%
Then%
\begin{equation*}
f\left( y_{1},y_{2},y_{3}\right) =\frac{\left( 1-y_{1}\right) ^{2}}{y_{1}^{2}%
}+\frac{\left( 1-y_{2}\right) ^{2}}{y_{2}^{2}}+\frac{\left(
3-y_{1}-y_{2}\right) ^{2}}{\left( y_{1}+y_{2}\right) ^{2}}\text{,}
\end{equation*}%
and $f$ is strictly convex on $\left\{ \left( y_{1},y_{2},y_{3}\right) \in 
\mathbb{R}_{+}^{3}\mid 2y_{1}/3<1\text{, }2y_{2}/3<1\text{,}2\left(
y_{1}+y_{2}\right) /3<5/3\right\} $, and strictly concave on $\left\{ \left(
y_{1},y_{2},y_{3}\right) \in \mathbb{R}_{+}^{3}\mid 2y_{1}/3>1\text{, }%
2y_{2}/3>1\text{,}2\left( y_{1}+y_{2}\right) /3>5/3\right\} $ (the last
condition is not active). Of course, $f$ may be strictly convex or concave
on points outside these sets also;\ these conditions are sufficient, but not
necessary. The domains of local convexity and local concavity are plotted in
the next figure.

\begin{figure}[H]
\centering
\includegraphics[scale=0.25]{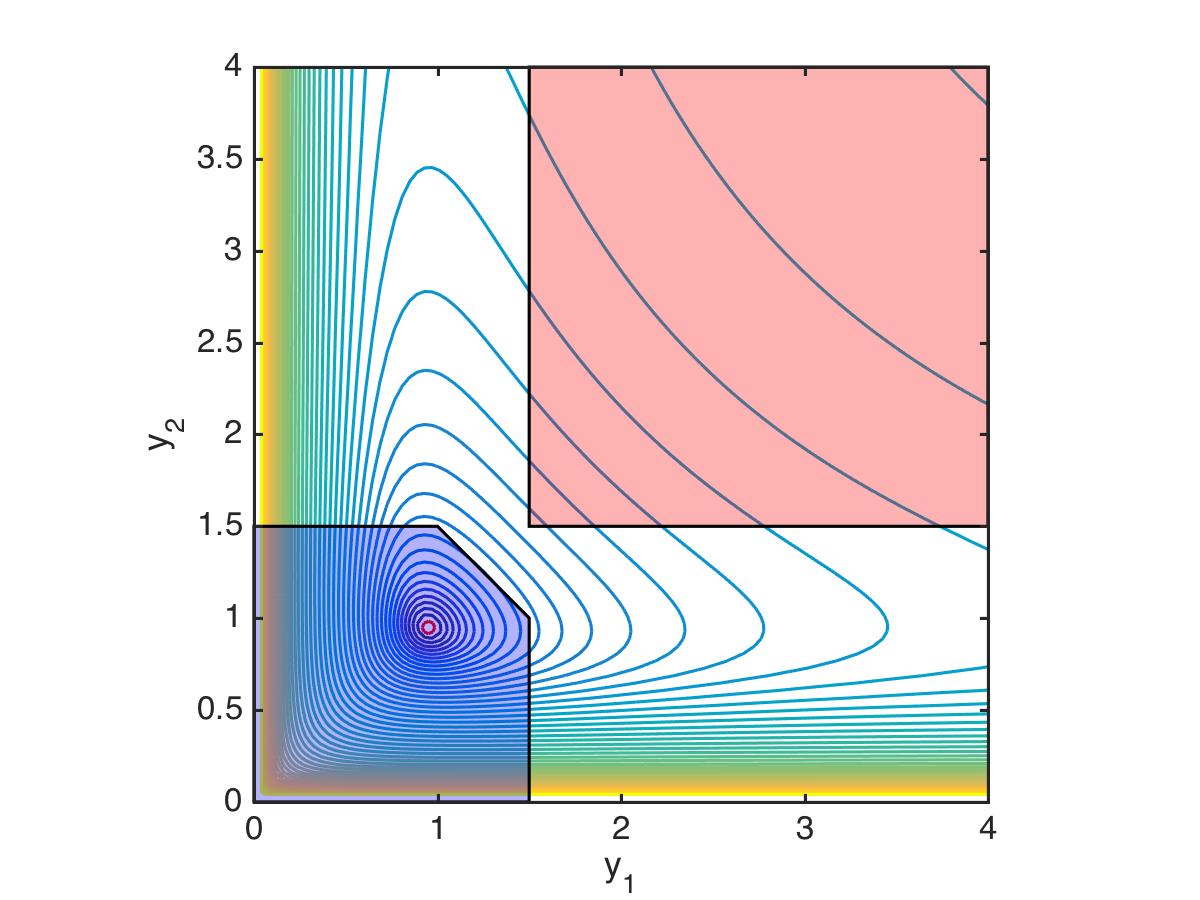}
\caption{This
figure shows isocurves of $f\left( y_{1},y_{2}\right) $, the local convexity
region (blue) and the local concavity region (red), determined from Theorem
1. The minimum point is marked with a red circle.}
\label{fig2}
\end{figure}

\end{example}

\begin{remark}
When $\sigma $ is rescaled\ (i.e., replacing $\sigma \left( \mu \right) $ by 
$\sigma _{\delta }\left( \mu \right) $), the derivatives are scaled
according to $\sigma _{\delta }^{\prime }\left( \mu \right) =\sigma ^{\prime
}\left( \mu /\delta \right) $ and $\sigma _{\delta }^{\prime \prime }=\delta
^{-1}\sigma ^{\prime \prime }\left( \mu /\delta \right) $, and hence $\alpha
\left( \mu \right) $ is replaced by $\alpha _{\delta }\left( \mu \right)
=\delta ^{-1}\alpha \left( \mu /\delta \right) $, and $\beta \left( \mu
\right) $ is replaced by $\beta _{\delta }\left( \mu \right) =\delta
^{-1}\beta \left( \mu /\delta \right) $.
\end{remark}

\subsection{Examples of follower problems with unique and non-unique
solutions}

In Theorem \ref{thm:local_convexity} we established sufficient conditions
for the objective function $f$ to be locally strictly convex with the view
of determining when the follower problem is well posed. Alas, strict
convexity of the objective function is not sufficient to make the follower
problem convex as the constraints have to be taken into account. We explore
the interplay between the objective function and the constraints and their
effect on minima and well posedness of the minimization problem through a
series of examples. It is instructive to begin with the linear case $\sigma
(\mu )=\mu $, and then proceed to the nonlinear cases (smooth threshold and
piecewise linear), compare (\ref{eqn:thesigmas}).

\subsubsection{The linear case}

In the linear case, the minimization problem for $f$ can be formulated as a
constrained minimization problem with a convex objective function. However,
the constraints are not convex, which can cause multiple minimum points for
certain values of $A$ and $\boldsymbol{z}$. For $\sigma \left( \mu \right)
=\mu $, $f$ is a strictly convex function of $\boldsymbol{\xi }\equiv 
\boldsymbol{\mu }^{-1}=\left( \mu _{1}^{-1},...,\mu _{m}^{-1}\right) $:%
\begin{equation*}
f\left( \boldsymbol{y}\right) =F\left( \boldsymbol{\xi }\right) \equiv
\sum_{i}\left( z_{i}\xi _{i}-1\right) ^{2}\text{.}
\end{equation*}%
Writing $F$ on the form%
\begin{equation*}
F\left( \boldsymbol{\xi }\right) =\sum_{i:b_{i}\neq 0}\frac{\left( \xi
_{i}-z_{i}^{-1}\right) ^{2}}{\left( z_{i}^{-1}\right) ^{2}}+m-m^{\prime }%
\text{,}
\end{equation*}%
where $m^{\prime }=\sum_{i:b_{i}\neq 0}1$ is the number of $b_{i}\neq 0$, we
see that for $c>0$, the level surface $F\left( \boldsymbol{\xi }\right)
=c+m-m^{\prime }$ is an $m^{\prime }$--axial ellipsoidal cylinder with
center coordinates $\xi _{i}=z_{i}^{-1}$ and corresponding semiaxes $\sqrt{c}%
z_{i}^{-1}$ for all $i$ such that $z_{i}\neq 0$, extending linearly along
all coordinates $\xi _{i}$ for which $z_{i}=0$. Let us assume that all $%
z_{i}\neq 0$ for simplicity. Then for $c>0$ the level surface $F\left( 
\boldsymbol{\xi }\right) =c$ is an $m$--axial ellipsoid centered at $\left(
z_{1}^{-1},...,z_{m}^{-1}\right) $ with semiaxes $\sqrt{c}z_{1}^{-1},...,%
\sqrt{c}z_{m}^{-1}$. The minimization problem for $f$ can be formulated as a
constrained minimization problem%
\begin{eqnarray*}
&&\min F\left( \boldsymbol{\xi }\right) \\
\boldsymbol{G}\left( \boldsymbol{\xi },\boldsymbol{y}\right) &\leq &%
\boldsymbol{0} \\
\boldsymbol{H}\left( \boldsymbol{\xi },\boldsymbol{y}\right) &=&\boldsymbol{0%
}
\end{eqnarray*}%
where $\boldsymbol{G}\left( \boldsymbol{\xi },\boldsymbol{y}\right) =-%
\boldsymbol{y}$, and $\boldsymbol{H}=\left( H_{1},...,H_{m}\right) $, where%
\begin{equation*}
H_{i}\left( \boldsymbol{\xi },\boldsymbol{y}\right) =\varphi _{i}\left( 
\boldsymbol{y}\right) -\xi _{i}\text{, }i=1,...,m
\end{equation*}%
and 
\begin{equation*}
\varphi _{i}\left( \boldsymbol{y}\right) =\frac{1}{\left( A\boldsymbol{y}%
\right) _{i}}\text{.}
\end{equation*}%
The objective function $F\left( \boldsymbol{\xi },\boldsymbol{y}\right)
=F\left( \boldsymbol{\xi }\right) $, is convex, and depends on $\boldsymbol{z%
}$, but is independent of $A$. The inequality constraints are linear and
hence convex. The equality constraints depend on $A$, but are independent of 
$\boldsymbol{z}$, and are not convex, so the problem is not convex. We have
thus separated the dependencies of $A$ and $\boldsymbol{z}$ into the
objective and constraint functions, respectively.

The equality constraints $\boldsymbol{H}\left( \boldsymbol{\xi },\boldsymbol{%
y}\right) =\boldsymbol{0}$ define an $n$--dimensional parametrized surface $%
S $ in $\mathbb{R}^{m}$ by $\boldsymbol{\xi }=\boldsymbol{\varphi }\left( 
\boldsymbol{y}\right) $. The $n$--dimensional tangent space to $S$ at $%
\boldsymbol{\varphi }\left( \boldsymbol{y}\right) $ is spanned by the
tangent vectors%
\begin{equation*}
\frac{\partial \boldsymbol{H}}{\partial y_{j}}\left( \boldsymbol{y}\right) 
\text{, }j=1,...,n\text{. }
\end{equation*}%
The Karush-Kuhn-Tucker (KKT) conditions (see e.g. \cite{NocedalWright2006})
for the minimization problem are%
\begin{eqnarray*}
-\frac{\partial F}{\partial \xi _{i}} &=&\sum_{k=1}^{n}\mu _{k}\frac{%
\partial G_{k}}{\partial \xi _{i}}+\sum_{l=1}^{m}\lambda _{l}\frac{\partial
H_{l}}{\partial \xi _{i}}=\lambda _{i}\text{, }i=1,...,m \\
0 &=&-\frac{\partial F}{\partial y_{j}}=\sum_{k=1}^{n}\mu _{k}\frac{\partial
G_{k}}{\partial y_{j}}+\sum_{l=1}^{n}\lambda _{l}\frac{\partial H_{l}}{%
\partial y_{j}}\text{, }j=1,...,n \\
\mu _{j}G_{j} &=&0\text{, }j=1,...,n \\
\mu _{j} &\geq &0\text{, }j=1,...,n
\end{eqnarray*}%
which are necessary conditions for minimum. Assume for simplicity that we
have a minimum with $\boldsymbol{y}^{\ast }>0$. Then $G_{j}<0$ so $\mu
_{j}=0 $. Substituting the first equation in the second we get 
\begin{equation*}
\sum_{l=1}^{m}\frac{\partial F}{\partial \xi _{l}}\frac{\partial H_{l}}{%
\partial y_{j}}=0\text{, }j=1,...,n\text{.}
\end{equation*}%
Geometrically, this conditions means that the gradient of $F$ is orthogonal
to the tangent space of $S$ at $y^{\ast }$. Hence, in the $\boldsymbol{\xi }$%
--space, minimal points $\boldsymbol{y}^{\ast }$ are characterized by points 
$\boldsymbol{\xi }^{\ast }=\boldsymbol{\varphi }\left( \boldsymbol{y}^{\ast
}\right) $ where the $n$--dimensional surface $S:\boldsymbol{\xi }=%
\boldsymbol{\varphi }\left( \boldsymbol{y}\right) $ is tangent to the
ellipsoid $F\left( \boldsymbol{\xi }\right) =c^{\ast }$. The optimal value
is then $c^{\ast }$.

Suppose now that $A$ is fixed. For $\boldsymbol{z}$, the center of the
ellipsoidal isosurfaces of $F$, let $E\left( \boldsymbol{z}\right) $ denote
the smallest isosurface of $F$ that intersects $S$. For generic $\boldsymbol{%
z}$, the ellipsoid $E\left( \boldsymbol{z}\right) $ will contain only one
point $\boldsymbol{\varphi }\left( \boldsymbol{y}^{\ast }\right) $, defined
by the unique minimum point $\boldsymbol{y}^{\ast }$. However, for $%
\boldsymbol{z}$ in some exceptional lower--dimensional set, $E\left( 
\boldsymbol{z}\right) $ will touch $S$ at several points, in which case we
have several minimum points $\boldsymbol{y}^{\ast }$. In that case, $A$ or $%
\boldsymbol{z}$ can be perturbed so that either of the multiple minima is
perturbed into a single global minimum. Hence the global minimum point
varies discontinuously with $A$ and $\boldsymbol{z}$, and the problem is ill
posed (unless $\boldsymbol{z}$ is restricted to be generic, e.g., if $%
\boldsymbol{z}^{-1}$ is sufficiently close to the surface $S$).

\begin{example}
\label{ex:linear}Let%
\begin{equation*}
A=%
\begin{bmatrix}
1 & \varepsilon \\ 
\varepsilon & 1 \\ 
1 & 1%
\end{bmatrix}%
\text{, }\boldsymbol{z}=%
\begin{bmatrix}
1 \\ 
1 \\ 
t%
\end{bmatrix}%
\end{equation*}%
where $\varepsilon =0.1$.

\begin{figure}[H]
\centering
\includegraphics[scale=0.15]{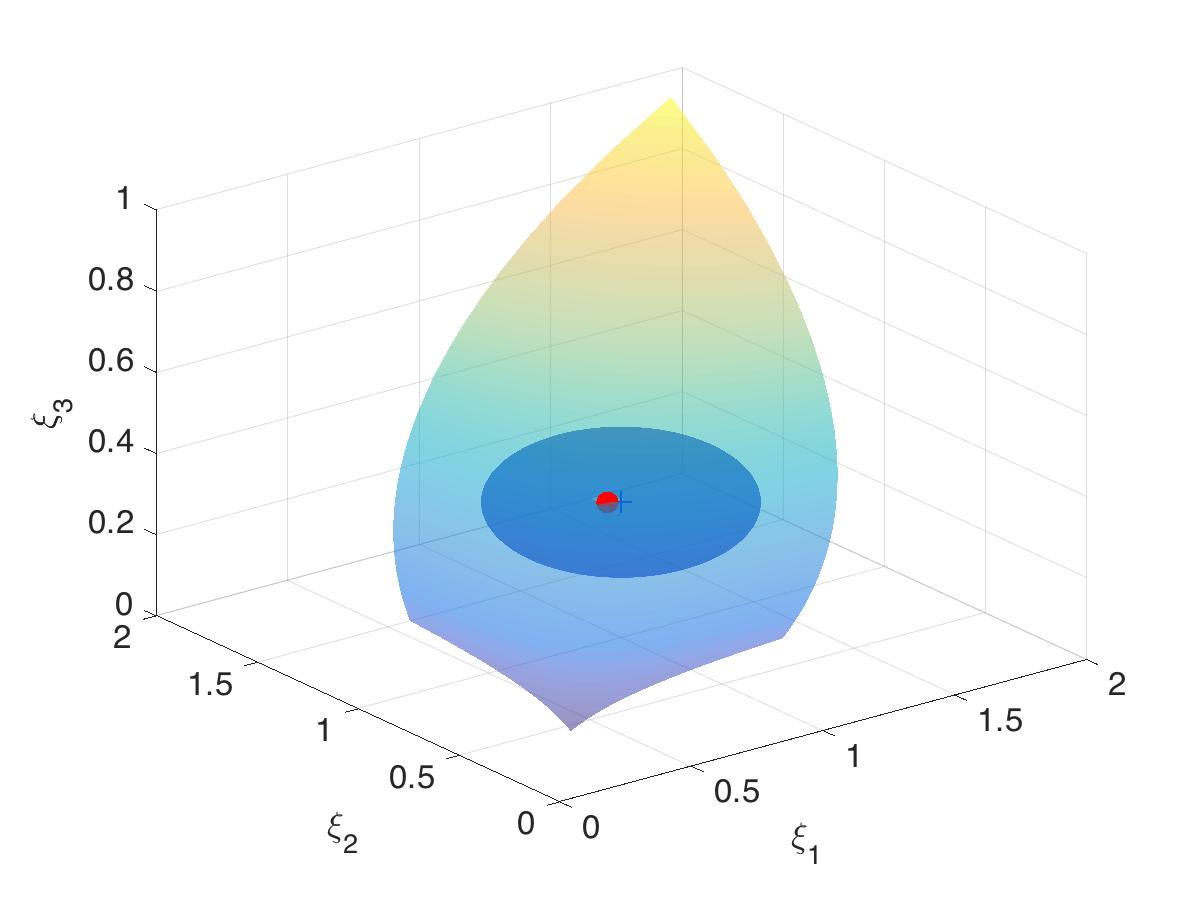}
\caption{Case $%
t=3$. The figure shows minimum point $\protect\xi ^{\ast }=\left(
0.7729,0.7729,0.4251\right) $ as a red dot, the surface $S$ parametrized for 
$0.5\leq y_{1},y_{2}\leq 5$, and the ellipsoidal isosurface for $f$ with
optimal value $c^{\ast }=0.17893$. The center $\boldsymbol{z}$ of the
ellipsoid is marked by a blue cross.}
\label{fig3}
\end{figure}

\begin{figure}[H]
\centering
\includegraphics[scale=0.15]{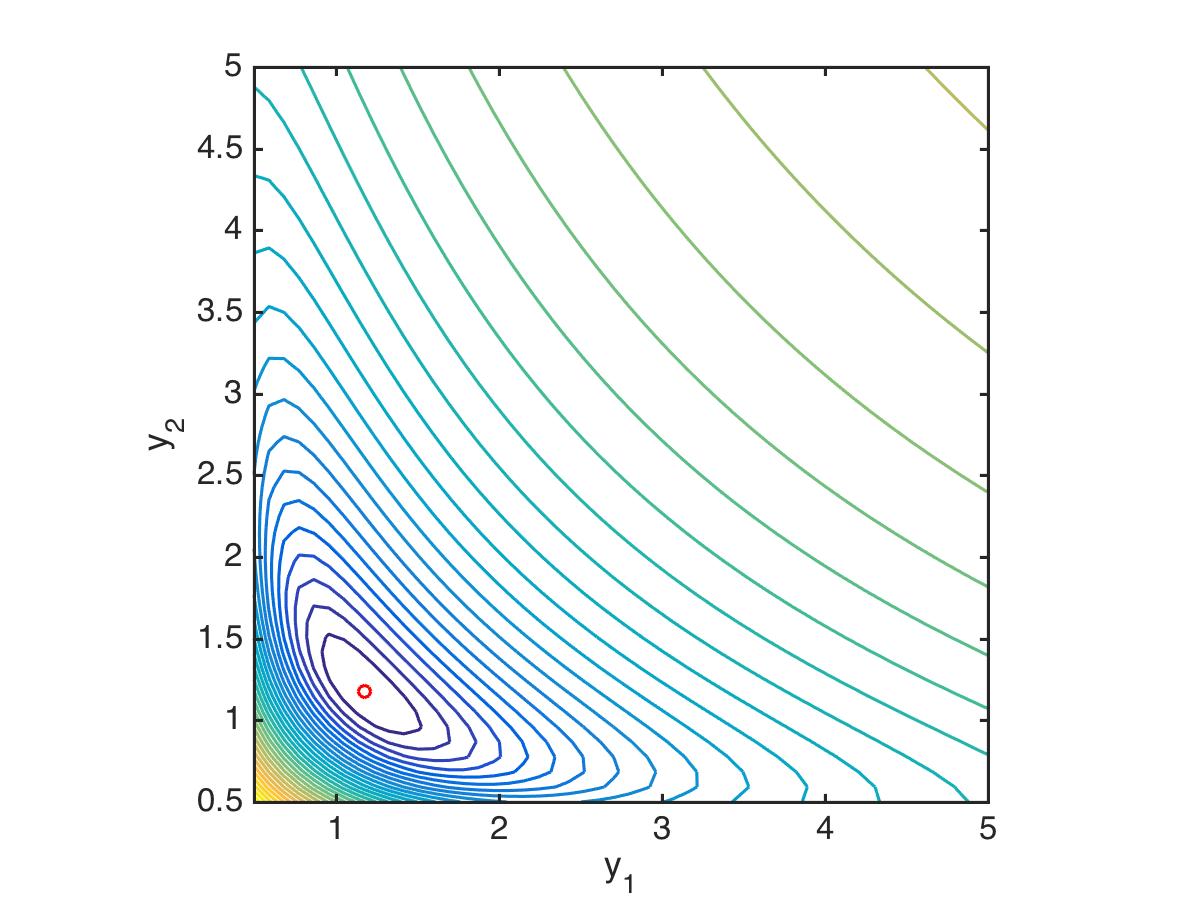}
\caption{Case $t=3$. The figure
shows the minimum point $y^{\ast }=\left( 1.1762;1.1762\right) $ marked with
a red circle, with the minimal value $c^{\ast }=0.17893$. Also isocurves for 
$f$ are shown.}
\label{fig4}
\end{figure}

\begin{figure}[H]
\centering
\includegraphics[scale=0.15]{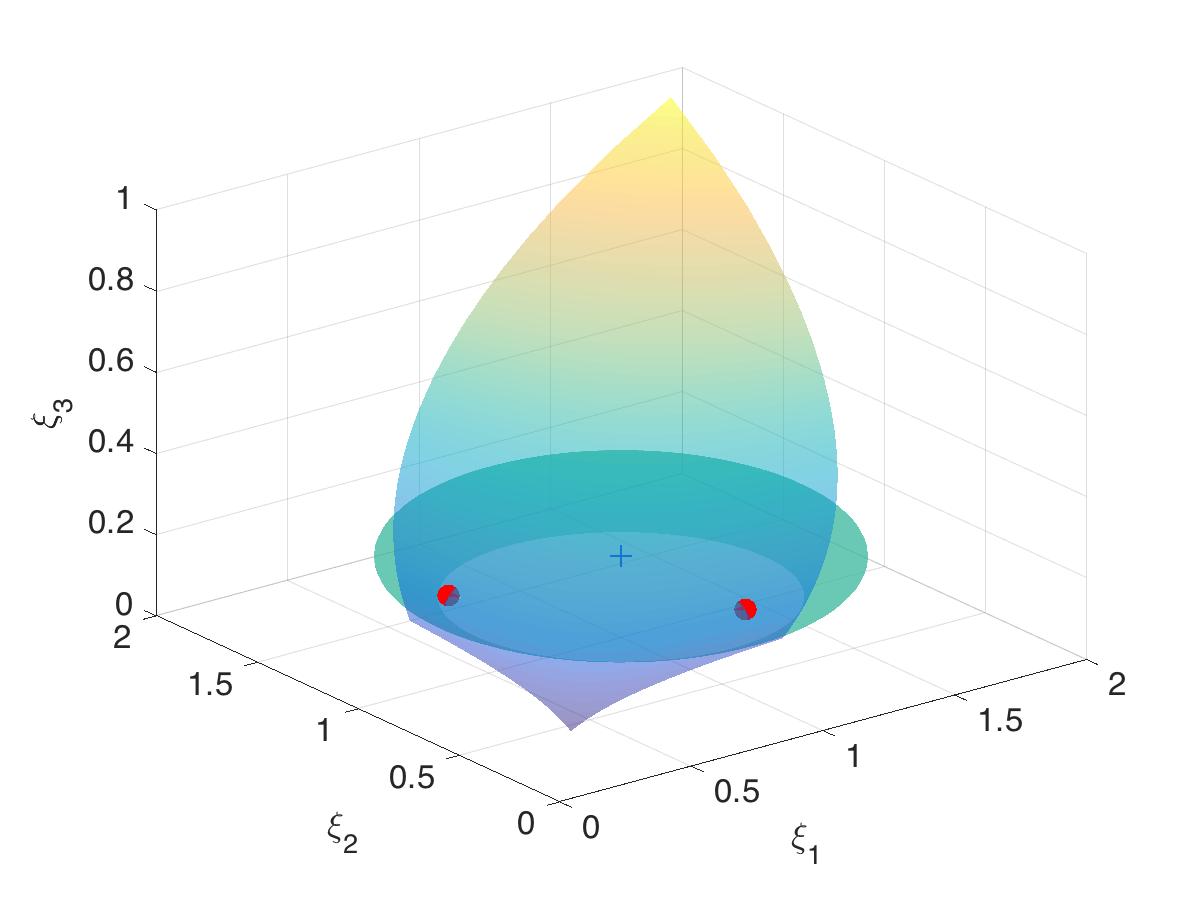}
\caption{Case $t=5.$, two minimum points $\protect\xi _{1}^{\ast }=\left(
0.9303,0.2920,0.2444\right) $, $\protect\xi _{2}^{\ast }=\left(
0.2920,0.9303,0.2444\right) $ marked by red dots, parametric surface $S$
plotted for $0.5\leq y_{1},y_{2}\leq 5$, and ellipsoidal isosurface for $f$
\ and optimal value $c^{\ast }=0.55556$. The center $\boldsymbol{z}$ of the
ellipsoid is marked by a blue cross.}
\label{fig5}
\end{figure}

\begin{figure}[H]
\centering
\includegraphics[scale=0.15]{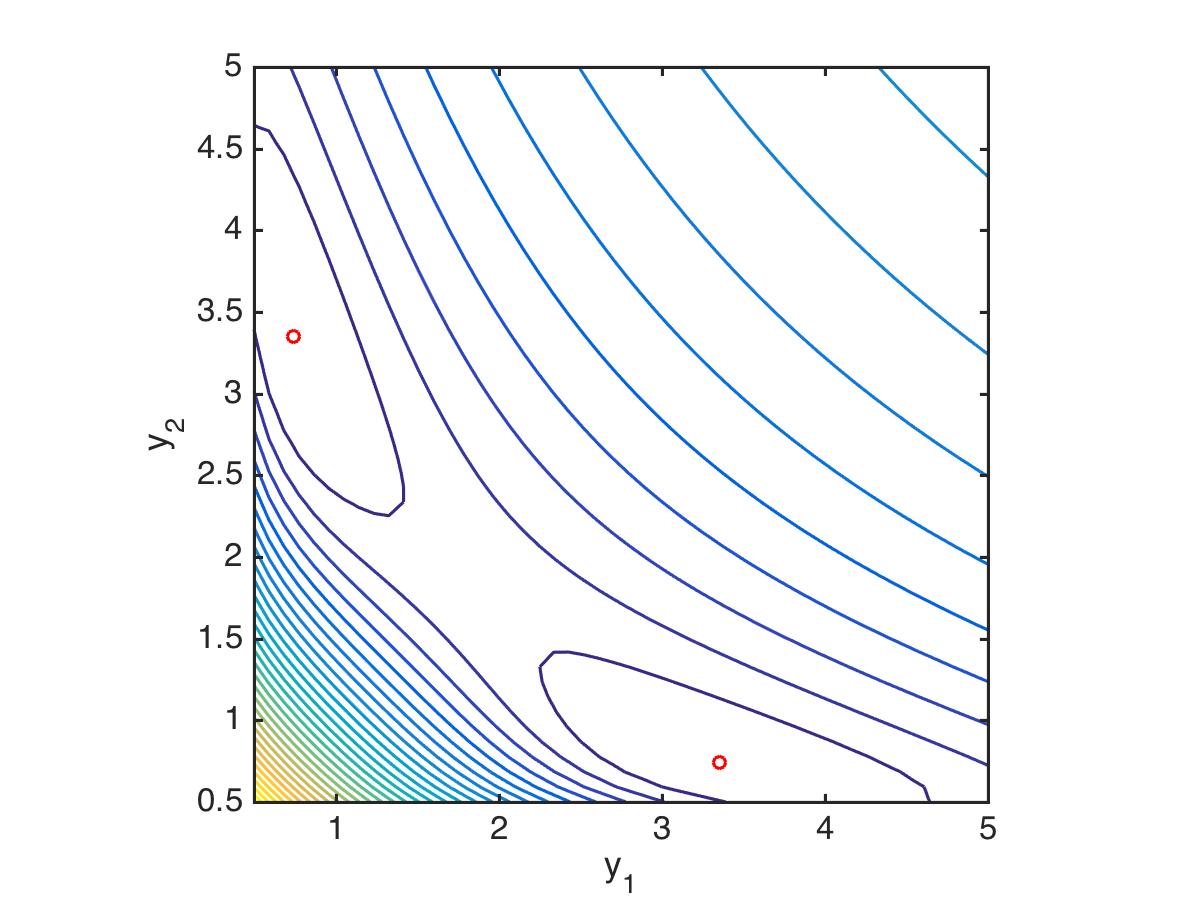}
\caption{ Case $t=5$. Two
minimal points $y_{1}^{\ast }=\left( 0.73987,3.351\right) $, $y_{2}^{\ast
}=\left( 3.351,0.73987\right) $, marked by red circles. Minimum value $%
c^{\ast }=0.55556$. Isosurfaces for $f$ also shown.}
\label{fig:linear_case_t5}
\end{figure}

\end{example}

\paragraph{The one--dimensional case}

Distinct minima cannot occur in the case $n=1$. We have

\begin{proposition}
If $a_{i}>0$, and some $b_{i}>0$ then 
\begin{equation*}
f\left( y\right) =\sum_{i=1}^{m}\frac{\left( b_{i}-a_{i}y\right) ^{2}}{%
\left( a_{i}y\right) ^{2}}=m-\frac{\alpha ^{2}}{\beta }+\beta \left( y^{-1}-%
\frac{\alpha }{\beta }\right) ^{2}
\end{equation*}%
where 
\begin{equation*}
\chi _{i}=b_{i}/a_{i}\text{, }\alpha =\sum_{i}\chi _{i}\text{ and }\beta
=\sum_{i}\chi _{i}^{2}\text{, }
\end{equation*}%
and $f$ has minimum value $m-\alpha ^{2}/\beta $, attained for%
\begin{equation*}
y=\frac{\beta }{\alpha }\text{,}
\end{equation*}%
which is a strict minimum. Moreover, the minimum point $y$ depends
continuously on $a_{i}$, $b_{i}$, so the minimization problem is well--posed.
\end{proposition}

\begin{proof}
The formula for $f$ is proved by elementary calculations. The final
conclusion follows from the observation that $f\left( y\right) $ is a second
order polynomial in $y^{-1}$.
\end{proof}

\begin{remark}
Note that $\partial \chi _{i}/\partial a_{i}=-a_{i}^{-1}\chi _{i}^{2}$ and $%
\partial \chi _{i}/\partial b_{i}=a_{i}^{-1}$, and the sensitivity of the
minimum point $y=y\left( a_{1},...,a_{m},b_{1},...,b_{m}\right) $ is given by%
\begin{equation*}
\frac{\partial y}{\partial a_{i}}=\frac{1}{a_{i}}\frac{\left( y-2\chi
_{i}\right) \chi _{i}^{2}}{\alpha }\text{, }\frac{\partial y}{\partial b_{i}}%
=\frac{1}{a_{i}}\frac{2\chi _{i}-y}{\alpha }
\end{equation*}%
so the minimum point $y$ is very sensitive for $a_{i}$ and $b_{i}$ when $%
a_{i}$ is small, provided that $y\neq 2\chi _{i}$.
\end{remark}

\subsubsection{The nonlinear case}

In contrast to the linear case, $f$ can have multiple local minima for $n=1$
in the nonlinear case. For certain values of $A$,$\boldsymbol{z}$, these
minimum values may be identical, hence distinct global minima. When this
occurs, the minimum problem becomes ill--posed, as the following example
shows.

\begin{example}
Let%
\begin{equation*}
A=%
\begin{bmatrix}
0.3 \\ 
0.8%
\end{bmatrix}%
\text{, }\boldsymbol{z}=%
\begin{bmatrix}
t \\ 
0%
\end{bmatrix}%
\text{, }\delta =1
\end{equation*}%
and 
\begin{equation*}
f\left( y\right) =\sum_{i=1}^{m}\frac{\left( z_{i}-a_{i}y\right) ^{2}}{%
\sigma ^{2}\left( a_{i}y\right) }\text{, }i=1,2\text{. }
\end{equation*}%
If $t=t_{0}\approx 1.0732$, then $f$ has two distinct global minimum points
for the smooth threshold $\sigma (\mu )$, (the smooth threshold is defined
in equation (\ref{eqn:thesigmas})), see Figure \ref{fig:globalmin_t10732}.
Moreover, for all $t$ in a neighborhood $U$ of $t_{0}$, $f$ has two distinct
local minima . For $t<t_{0}$ in $U$, the left local minimum of $f$ is
global, and for $t>t_{0}$ in $U$, the right local minimum of $f$ is global.
Hence when $t$ increases from $t=1.05$ to $1.10$, the global minimum point $%
y=y\left( t\right) $ of $f\left( y\right) $ jumps from $y\approx 0.4919$ to $%
y\approx 3.5733$ at $t=t_{0}$, see figures \ref{fig:globalmin_t105} and \ref%
{fig:globalmin_t110}. The function $f$ shows a similar behaviour for the
piecewise linear $\sigma $, with a slightly different value of $t_{0}$.%

\begin{figure}[H]
\centering
\includegraphics[scale=0.15]{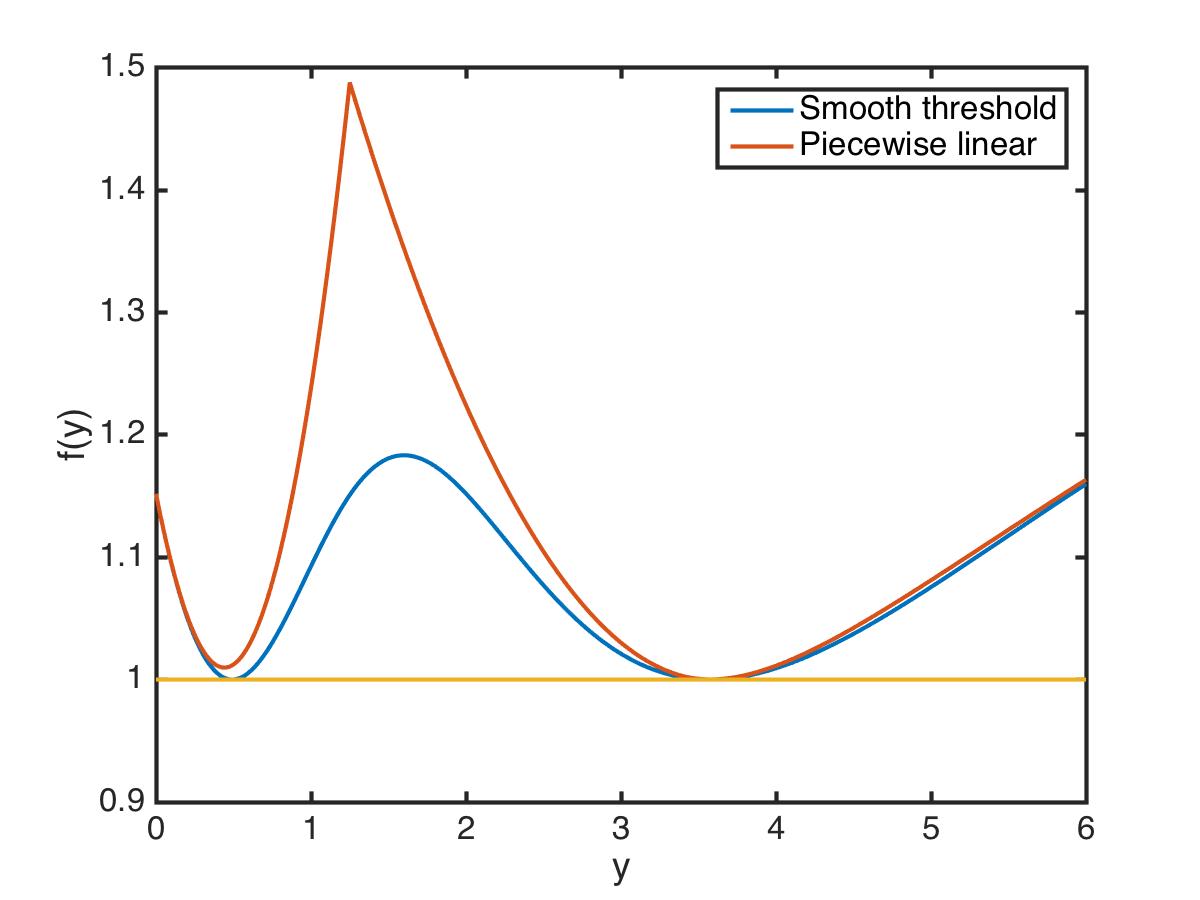}
\caption{ The figure shows $f\left(
y\right) $ for $t=t_{0}\approx 1.0732$. At this transition point, the smooth
threshold $\protect\sigma $ gives two distinct global minima, while the
piecewise linear $\protect\sigma $ gives a single global minimum to the
right. The transition point for the piecewise linear $\protect\sigma $
occurs for a sligthly different value of $t$.}
\label{fig:globalmin_t10732}
\end{figure}

\begin{figure}[H]
\centering
\includegraphics[scale=0.15]{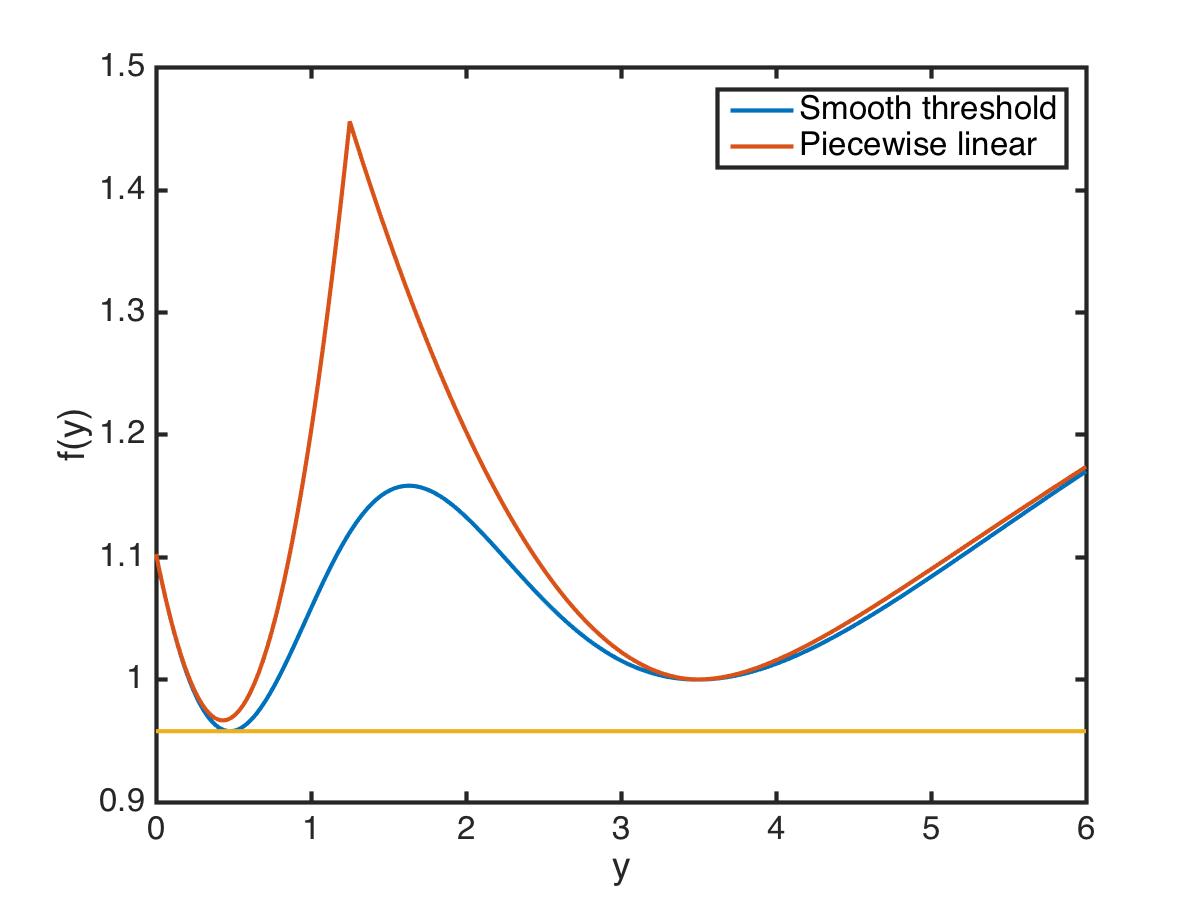}
\caption{ The figure shows $f\left( y\right) $ for $t=1.05$. Both smooth
threshold and piecewise linear $\protect\sigma $ give two distinct local
minima, the left one is also global minimum.}
\label{fig:globalmin_t105}
\end{figure}

\begin{figure}[H]
\centering
\includegraphics[scale=0.15]{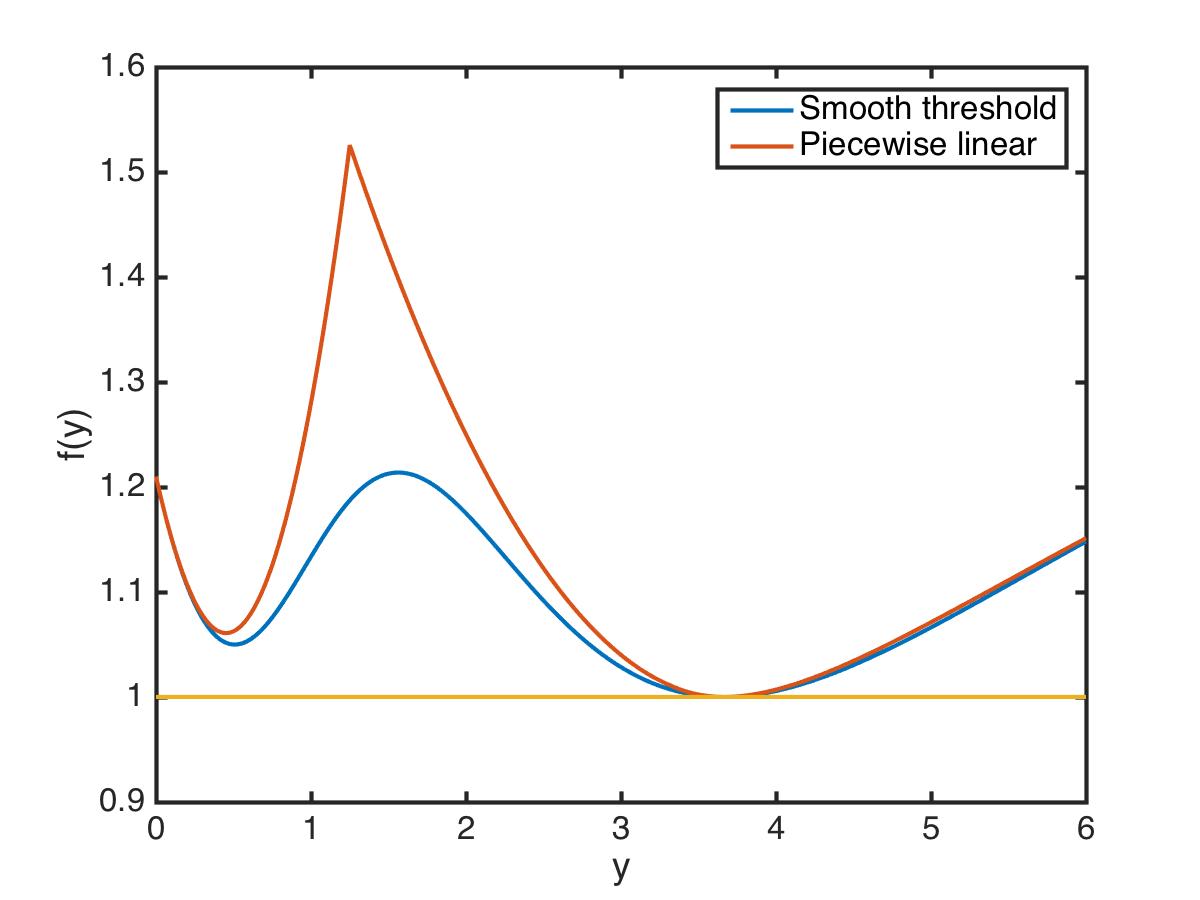}
\caption{ The figure shows $f\left( y\right) $ for $t=1.10$. Both smooth
threshold and piecewise linear $\protect\sigma $ give two distinct local
minima, the right one is also global minimum.}
\label{fig:globalmin_t110}
\end{figure}

\end{example}

In higher dimensions, nonlinear $\sigma $ give similar behaviour of $f$
regarding multiple global minima and ill--posedness of the minimization
problem.

In higher dimensions, the nonlinear $\sigma $ functions causes a similar
behaviour as the linear $\sigma $ regarding multiple local minimas of $f$.
To illustrate this, we show a slight perturbation of the bimodal example,
Example \ref{ex:linear}, above for the linear $\sigma $. The perturbed
problem still have two local minima, but only one of them is a global min.
Let 
\begin{equation*}
A=%
\begin{bmatrix}
1 & 0.1 \\ 
0.1 & 1 \\ 
1 & 1%
\end{bmatrix}%
\text{, }\boldsymbol{z}=%
\begin{bmatrix}
1.1 \\ 
0.9 \\ 
5%
\end{bmatrix}%
\text{.}
\end{equation*}%
Then we get


\begin{figure}[H]
\centering
\includegraphics[scale=0.15]{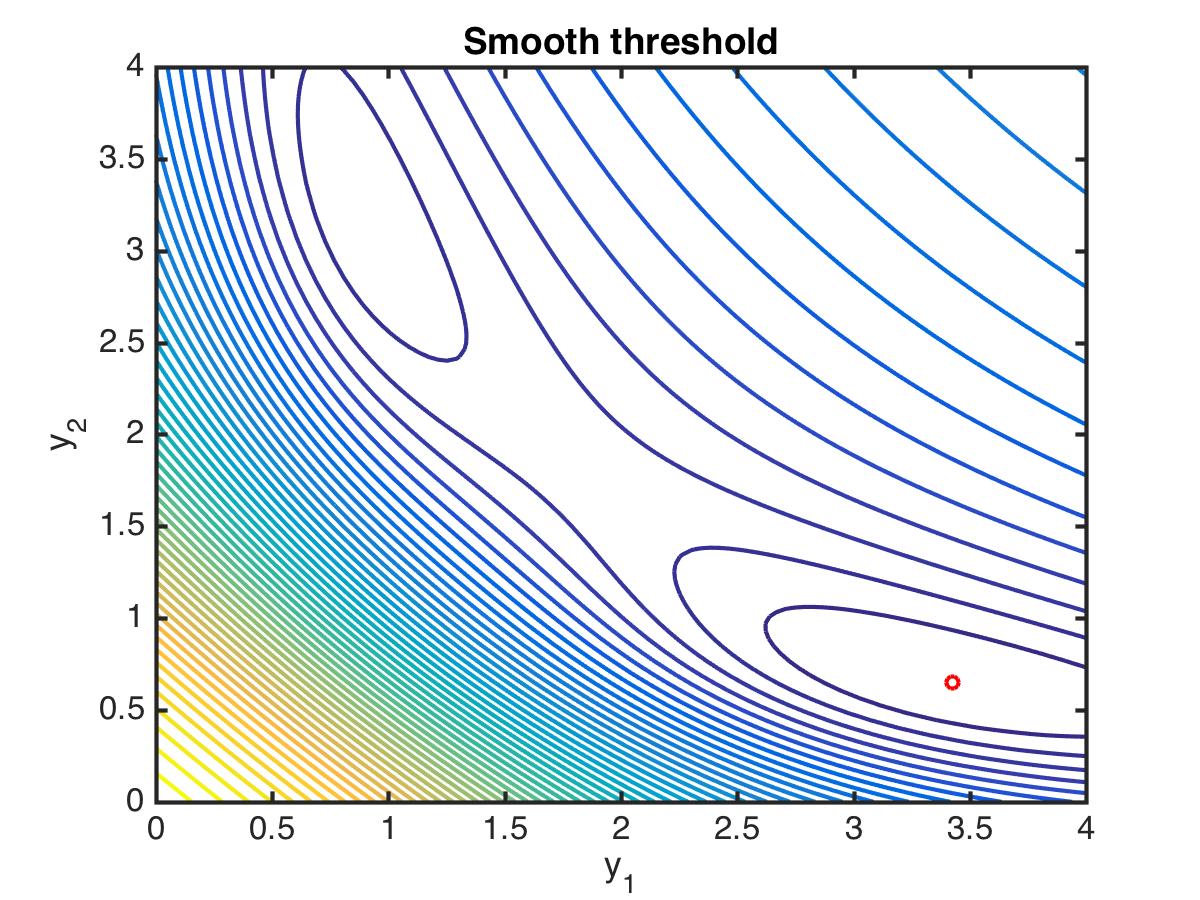}
\includegraphics[scale=0.15]{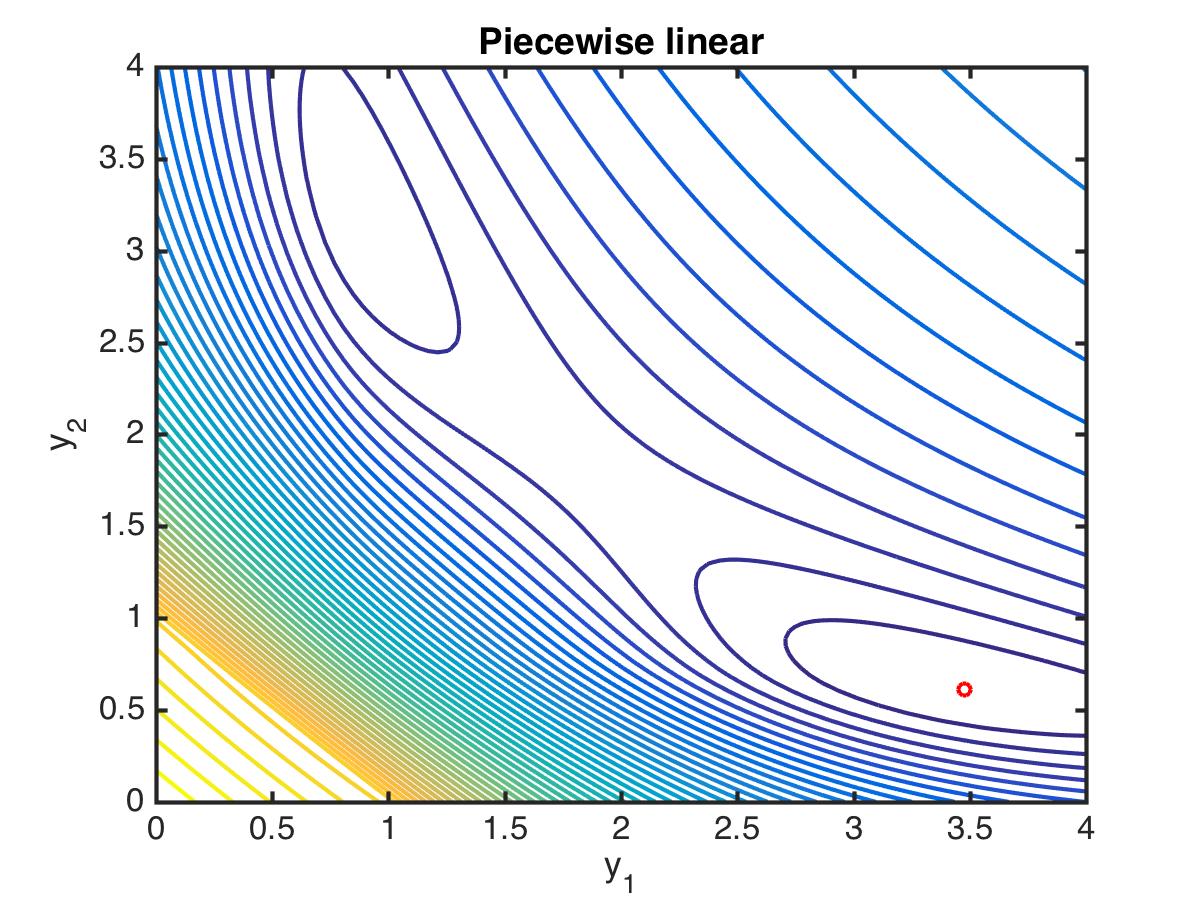}
\caption{Both smooth threshold and piecewise linear $\sigma $ functions give two
local minima in this case, and quite similar isocurve pattern. Comparison with the linear case, Figure %
\ref{fig:linear_case_t5}, in the previous example reveals a similar pattern.}
\label{fig:10_11}
\end{figure}



In some cases there are huge differences between smooth threshold and
piecewise linear $\sigma $. For example, by randomly choosing $A$ and $%
\boldsymbol{z}$, we found%
\begin{equation*}
A=%
\begin{bmatrix}
0.0557 & 0.8528 \\ 
0.3231 & 0.6868 \\ 
0.9336 & 0.1672%
\end{bmatrix}%
\text{, }\boldsymbol{z}=%
\begin{bmatrix}
1.3196 \\ 
0 \\ 
1.3499%
\end{bmatrix}%
\end{equation*}%
which gives

\begin{figure}[H]
\centering
\includegraphics[scale=0.15]{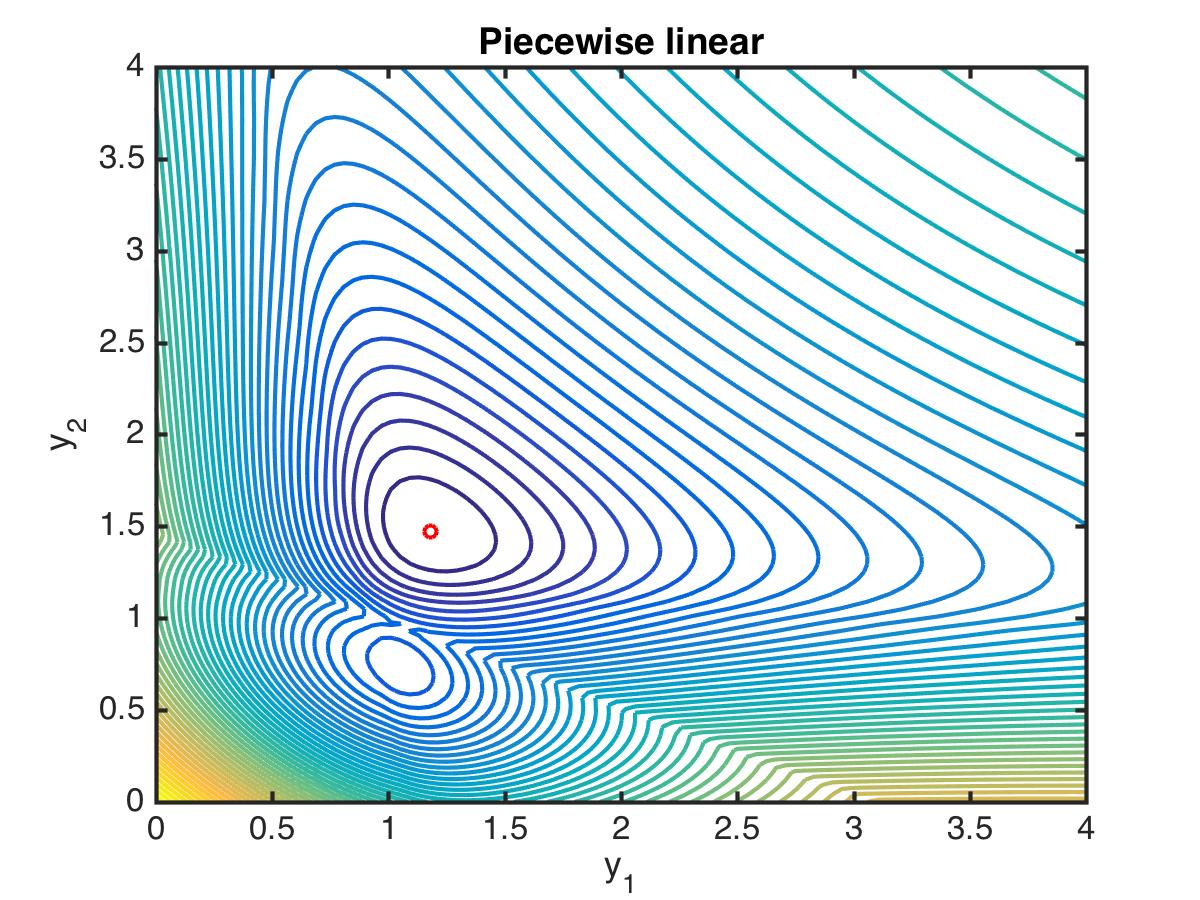}
\includegraphics[scale=0.15]{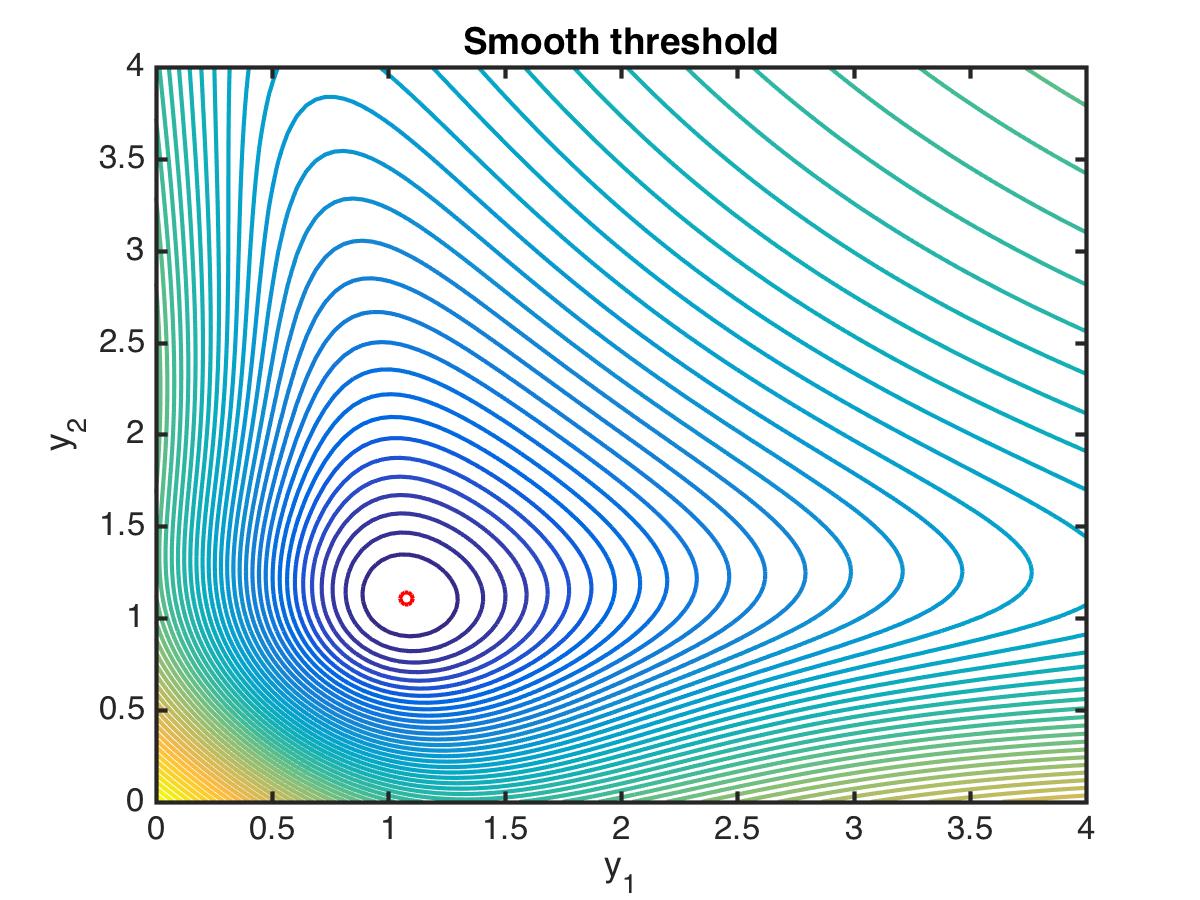}
\caption{Here we see a dramatic
difference between the discontinouous gradient and double local minima in
the piecewise linear case, and the smooth, single global minimum in the
smooth threshold and linear cases.}
\label{fig:qualitative_difference}
\end{figure}

A motivation for using a scaled piecewise linear $\sigma $ rather than the
linear $\sigma $ is that model values $\mu $ (averages) below a detection
limit should give a variance $\sigma $ of the measurement value $z$ which is
independent of $\mu $ (corresponding to the background noise). We think it
is a good compromise to use the smooth threshold $\sigma $, giving a
variance with a strictly positive lower bound, but avoiding potential
numerical difficulties as pictured in Figure \ref{fig:qualitative_difference}
above.

\section{Application: neutral gas in an urban environment}

We will now apply the bilevel optimisation method to experimental data
reconstructing the source term. The different choices of distance function,
smooth threshold, piecewise linear, linear and constant will be compared and
contrasted. The experimental data stems from the European Defence Agency
category B project MODITIC where wind tunnel experiments were conducted
involving release of both neutrally buoyant gas as well as dense gas over
urban environments of varying complexity. In all cases the wind tunnel
boundary layer was neutrally stable. The gas was released from a single
point source either as a puff or continuously at a constant rate. For the
full details on the experimental set up we refer to \cite{RobinsEtAl2016a}
and \cite{RobinsEtAl2016b}.

\subsection{Sensor data}

As the bilevel optimisation method described above is designed for linear
inverse dispersion problems we only consider cases where the released gas is
neutrally buoyant, and to further reduce the scope we focus on two scenarios
referred to as the simple array, a symmetric case with four buildings, and
the complex array, a less symmetric case involving fourteen buildings.

\subsubsection{The simple array}

The dimensions of the simple array, with the positions of the synchronized
detectors we use for backtracking is shown in Figure \ref{fig:SimpleArray}.%

\begin{figure}[H]
\centering
\includegraphics[scale=0.15]{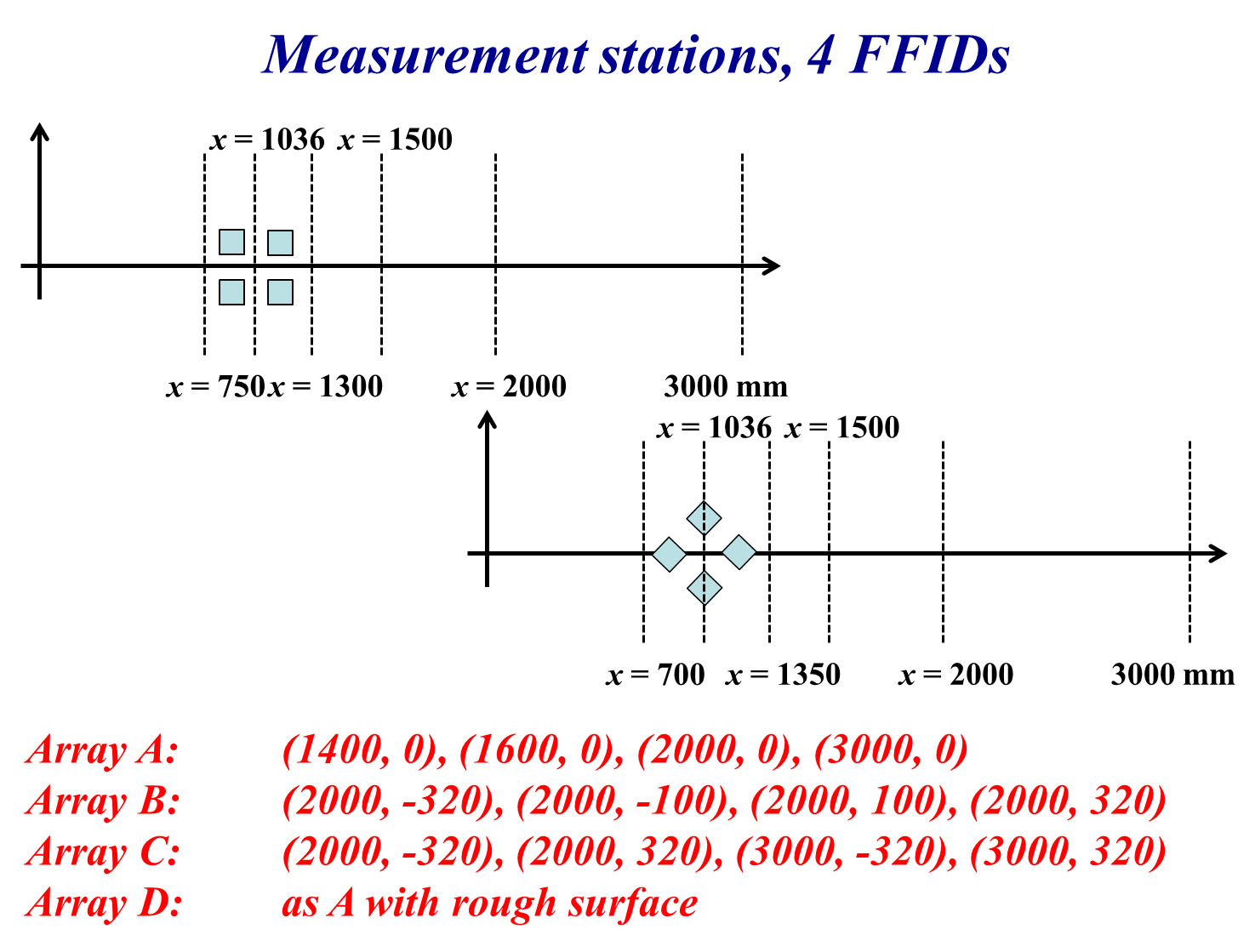}
\caption{Sensors network (A, B, C
arrangement) for the simple array cases. In each scenario 4 synchronized
detectors are used to measure the concentration of the released gas. The
dectors are located at the positions stated under A, B and C respectively.}
\label{fig:SimpleArray}
\end{figure}

The wind direction in the wind tunnel experiments is aligned with the x-axis
in Figure \ref{fig:SimpleArray}, and as shown in the figure the simple array
may aligned at two different angles: we denote the alignment in the upper
pane as \textquotedblleft 0 degrees\textquotedblright\ and the alignment in
the lower pane as \textquotedblleft 45 degrees\textquotedblright . In each
scenario four synchronized detectors were used, hence the reference 4 FFID
in Figure \ref{fig:SimpleArray}, and we chose three different sets of
locations for these detectors: we refer to these as case A, B and C. In
addition to this two different source locations were available: these we
denote S1 and S2, both S1 and S2 (at the origin) are located at 8H upwind
(=-0.88m) in the x direction, but S1 is shifted off the x-axis by 1.5H
(=0.165m) in the +y direction. The location of S2 was chosen to be the
origin of the coordinate system. The diameter of the sources is 0.1m (to be
compared to building height and sides of H=0.11m). A constant release rate
of 50l/min=8.33e-4 m$^{3}$/s was used in all scenarios.

\subsubsection{The complex array}

In the complex array there more buildings present and there is less
symmetry. The configuration also opens up for a larger number of sensible
detector locations, however, still only 4 synchronized detectors were used
in these scenarios. For the complex array we chose five different sets of
detector configurations: these are denoted case A through to E, and these
are shown together with the geometry of the complex array in Figure \ref%
{fig:ComplexArray}. Note that some detectors are located on the roof tops
(cases A and C) and other are inside the street canyons (cases B, D and E).

\begin{figure}[H]
\centering
\includegraphics[scale=0.35]{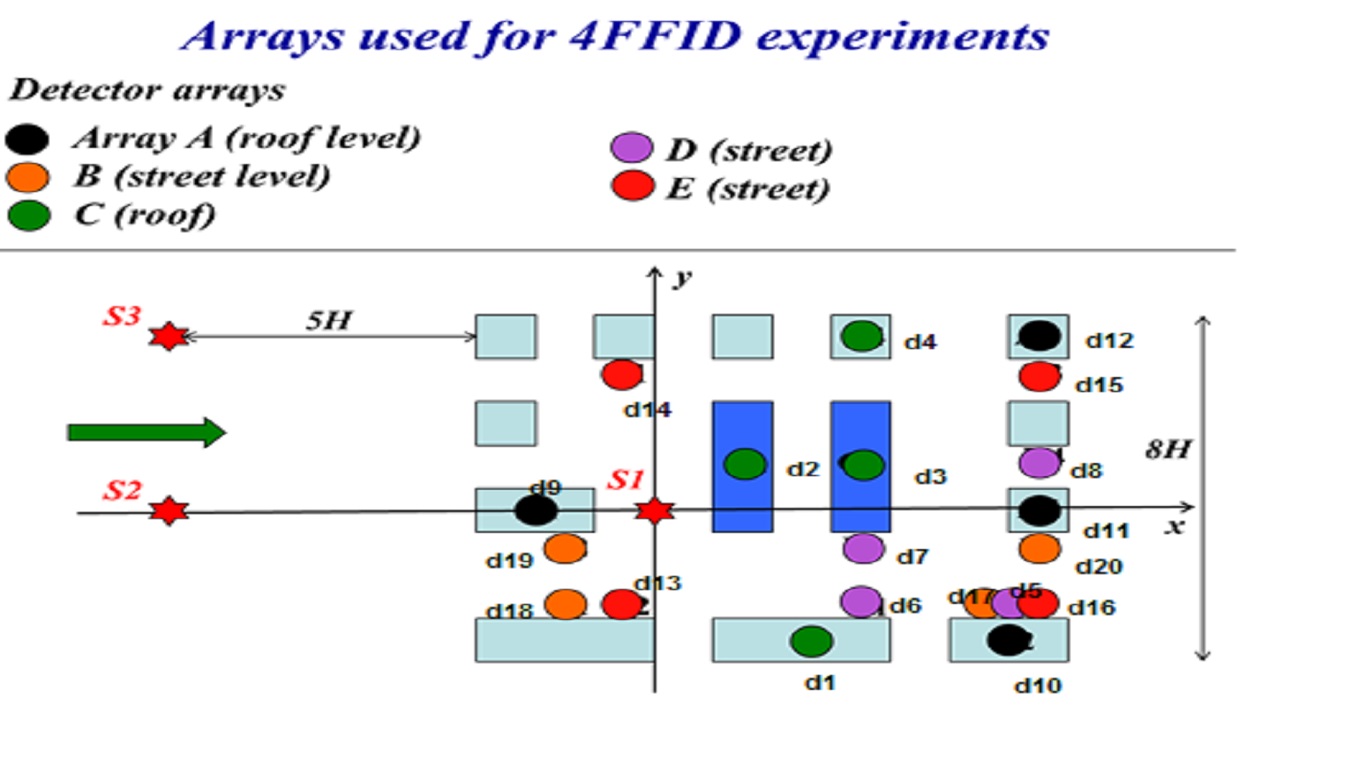}
\caption{Configuration of the complex
array. In each scenario, denoted case A through to E, 4 synchronized
detectors were used to measure the concentration of the released gas. Note
that some detectors are located at the roof of the buildings. Two different
source locations were used, denoted S1 and S2. (There is a third source S3
indicated in the figure, but it was not used for inverse modelling).}
\label{fig:ComplexArray}
\end{figure}

Two different source locations were used S1 and S2. The diameter of the
sources is 0.1m. The source S1 is defined to be at the origin (see Figure %
\ref{fig:ComplexArray}) while S2 is located upwind at x=-8H=-0.88m and y=0.
S1 and S2 are both located on the ground. As for the empty and complex array
the source strength is 50l/min=8.33e-4m$^{3}$/s.

\subsection{Source-sensor relationship: adjoint CFD-plumes}

To apply the bilevel optimisation method to the sensor data retrieved from
the wind tunnel we need a relationship between the source and the sensors,
i.e. a dispersion model. Since we are in a setting where a) we are
considering an urban environment and b) we are trying to gauge the fidelity
of the inverse solver, we opt to use a computational fluid dynamics solver
for the wind field. We are only considering the steady state of the
dispersion scenario (not the transient behaviour when the source is turned
on and subsequently off), hence a Reynolds-averaged Navier-Stokes (RANS)
solver is suitable. See \cite{BurkhartBurman} for RANS solutions of MODITIC
scenarios. The RANS solution in each case defines a source-sensor
relationship which could be used to solve the inverse problem, however, due
to computational efficiency it is preferable to solve the adjoint dispersion
problem \cite{Marchuk1986}. The dispersion problem is self-adjoint, hence
the adjoint solution is obtained by reversing the advection while leaving
the diffusive component untouched. The RANS solver in PHEONICS was adapted
accordingly to yield adjoint RANS solutions for the simple array and the
complex array \cite{WP7000}. One adjoint wind field was computed for each
sensor. These adjoint flow fields constitute a source-sensor relationship.

\subsection{Source reconstruction}

We will briefly outline the bilevel optimisation algorithm before using it
to reconstruct the MODITIC sources.

\subsubsection{Bilevel optimisation algorithm}

In each case four synchronous sensors where used, let us number them $%
i=1,2,3,4$. For each sensor $i$ compute the adjoint dispersion plume $\chi
_{i}(x,y,z)$ for each grid point $(x,y,z)$. Using the adjoint plume the
model sensor data is given by $c_{i}=q\chi _{i}(x,y,z)$ where $q$ is the
source strength of a point source located at grid point $(x,y,z)$. We denote
the sensor data obtained in the wind tunnel $d_{i}$.

\begin{enumerate}
\item Choose the weight $\sigma $ in the distance function. Choose detection
threshold $\delta $ (may be known from the sensor manufacturer). Choose
penalty coefficient $\lambda $.

\item Follower problem: Solve for the optimal emission rate $q^{\ast
}(x,y,z) $ in each grid point $(x,y,z)$ by%
\begin{equation*}
\min_{q}\sum_{i=1}^{4}\frac{\left( d_{i}-c_{i}^{{}}\right) ^{2}}{\sigma
^{2}\left( c_{i}^{{}}\right) }.
\end{equation*}

\item Compute the envelope grid function $c_{i}^{\ast }=q^{\ast }(x,y,z)\chi
_{i}(x,y,z)$ for each $(x,y,z)$ in the grid.

\item Leader problem: Minimize the grid function%
\begin{equation*}
V(x,y,z)=\min_{(x,y,z)}\left( \sum_{i=1}^{4}\frac{\left( d_{i}-c_{i}^{\ast
}\right) ^{2}}{\sigma ^{2}\left( c_{i}^{\ast }\right) }+\lambda q^{\ast
}(x,y,z)\right) ,
\end{equation*}%
,where $\lambda q^{\ast }$ is a penalty term, to find the optimal source
location $(x,y,z)=\left( x^{\ast },y^{\ast },z^{\ast }\right) .$

\item The optimal solution is given by $\left( q^{\ast },x^{\ast },y^{\ast
},z^{\ast }\right) $.
\end{enumerate}

\subsubsection{Results}

Below we present the source reconstruction results by showing the estimation
error in the $x,y$ and $z$ directions (source location) and in the emission
rate respectively. As the estimated emission rate varies over a large range
we present the error in the $\log _{10}$ value of the emission rate. We
consider four different choices of weight $\sigma $ in the distance
function: smooth threshold, piecewise linear, linear (as defined in equation
(\ref{eqn:thesigmas})) and in addition a constant weight $\sigma \equiv 1$.
The error for each parameter for each choice of $\sigma $ for each data set
in the simple and complex arrays are presented in scatter plots for
comparison. These results are presented for two different choices of high
emission rate penalty factor $\lambda $, $\lambda =10^{-2}$ and $\lambda =1$%
. In addition the results are presented for two different flavours of the
inverse problem: the "unconstrained" one where the source could be located
anywhere on or above the ground, and the "conditional" case where the source
is assumed (correctly!) to be on the ground.

To be specific, $\sigma \equiv 1$ in the constant case. In the other cases,
a scaled $\sigma $ is used: $\sigma _{\delta }\left( \mu \right) =\delta
\sigma \left( \mu /\delta \right) $. In the probabilistic interpretation, $%
\delta $ may be thought of as a detection limit. When the model value $\mu $
is below the detection limit, the standard deviation of the measured value $%
z $ is exactly $\delta $ in the piecewise linear case, and approximately $%
\delta $ in the smooth case. We use $\delta =10^{-6}$ in this study (a
choice of $\ \delta =10^{-3}$ gives no noticeable difference).

\subsection{Results for $\protect\lambda =10^{-2}$, $\protect\delta =10^{-6}$%
}

Setting $\lambda =10^{-2}$ and $\delta =10^{-6}$ yield the following results.

\subsubsection{3D estimates}


\begin{figure}[H]
\centering
\includegraphics[scale=0.15]{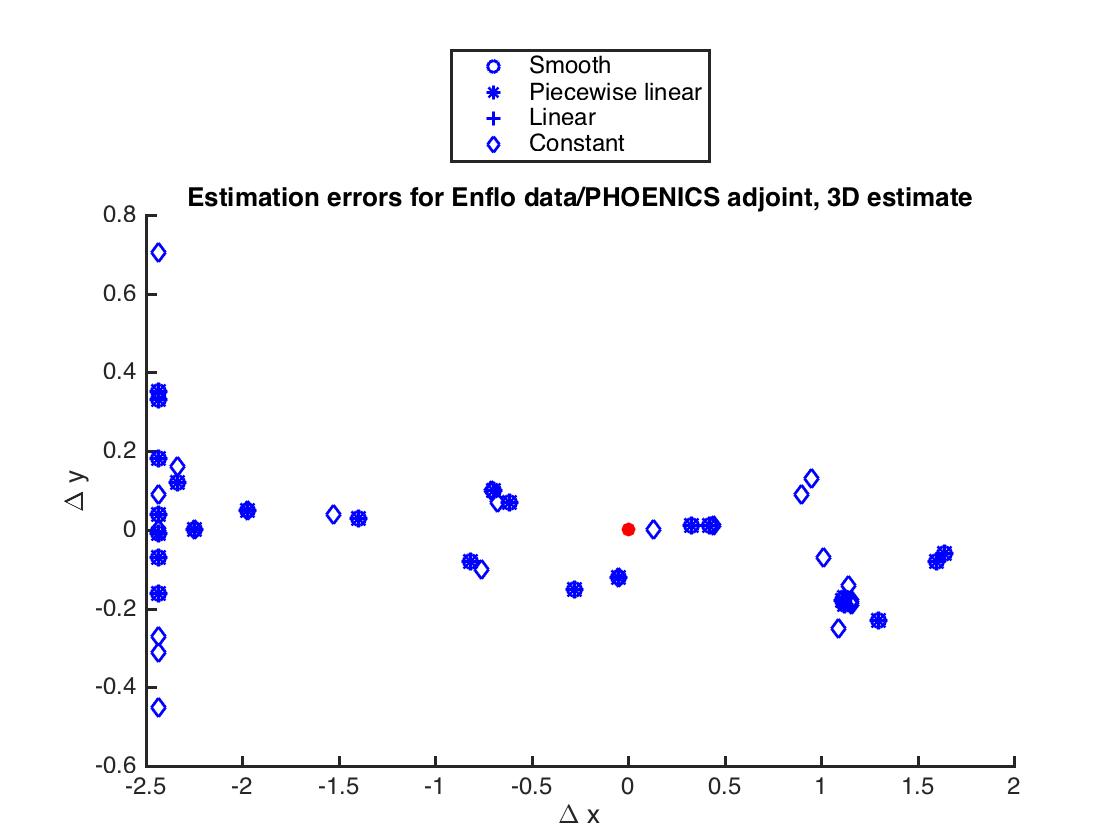}
\includegraphics[scale=0.15]{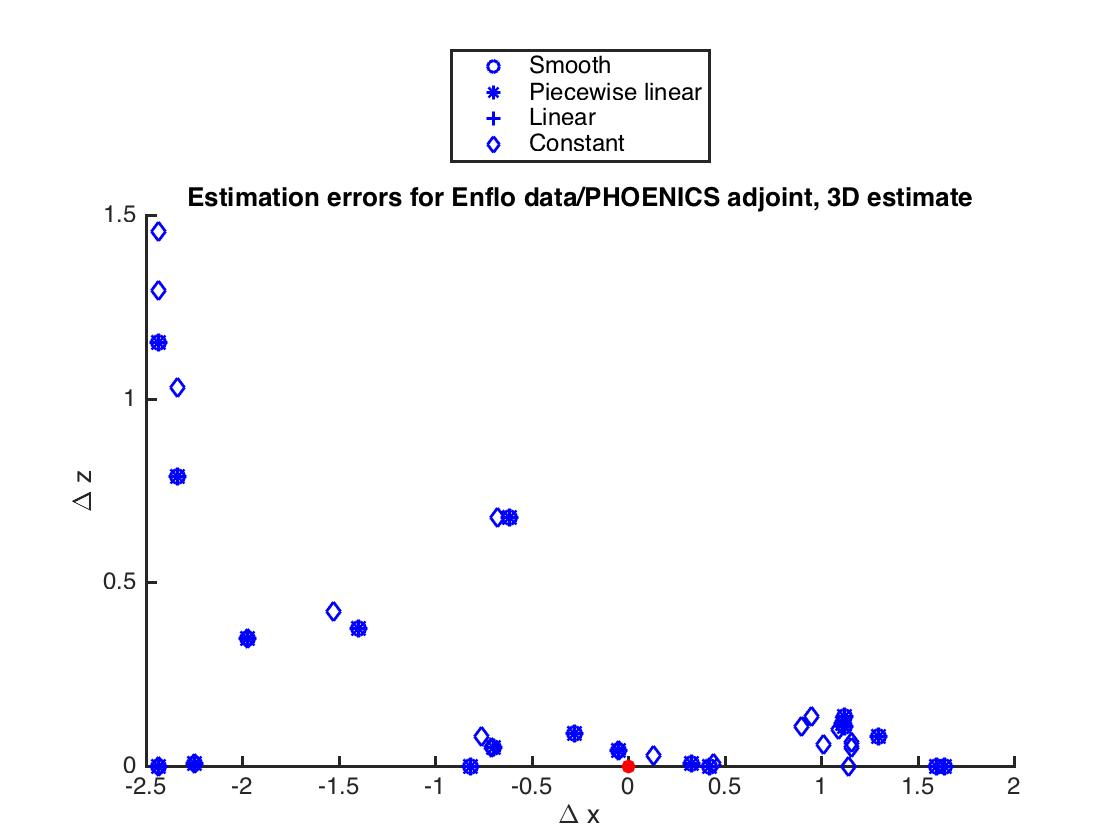}
\includegraphics[scale=0.15]{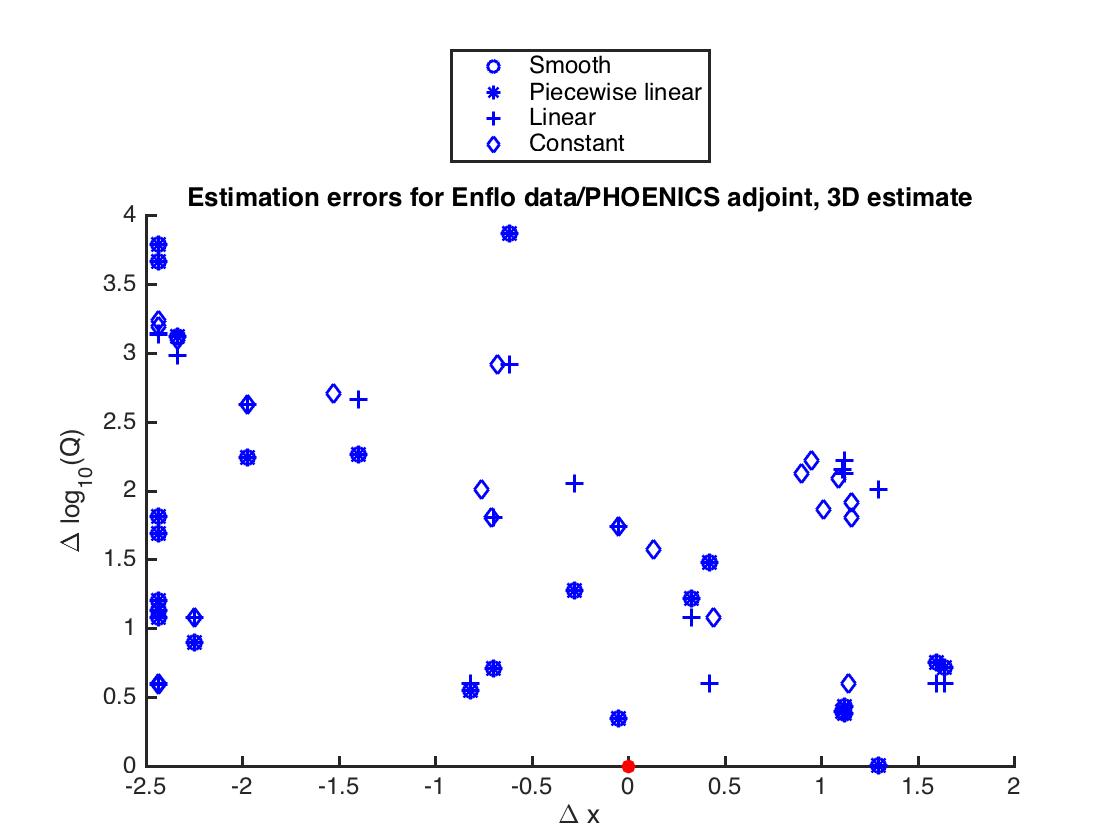}
\end{figure}


The average errors (averaged over all data sets) is presented in the
following table:

\begin{tabular}{lllll}
& Smooth & Piecewise linear & Linear & Constant \\ 
$x$ & $1.5$ & $1.5$ & $1.5$ & $1.42$ \\ 
$y$ & $0.117$ & $0.117$ & $0.117$ & $0.155$ \\ 
$z$ & $0.214$ & $0.214$ & $0.117$ & $0.251$ \\ 
$\log _{10}\left( q\right) $ & $1.46$ & $1.46$ & $1.63$ & $1.83$%
\end{tabular}

\subsubsection{2D estimates - assuming the source is on the ground}


\begin{figure}[H]
\centering
\includegraphics[scale=0.15]{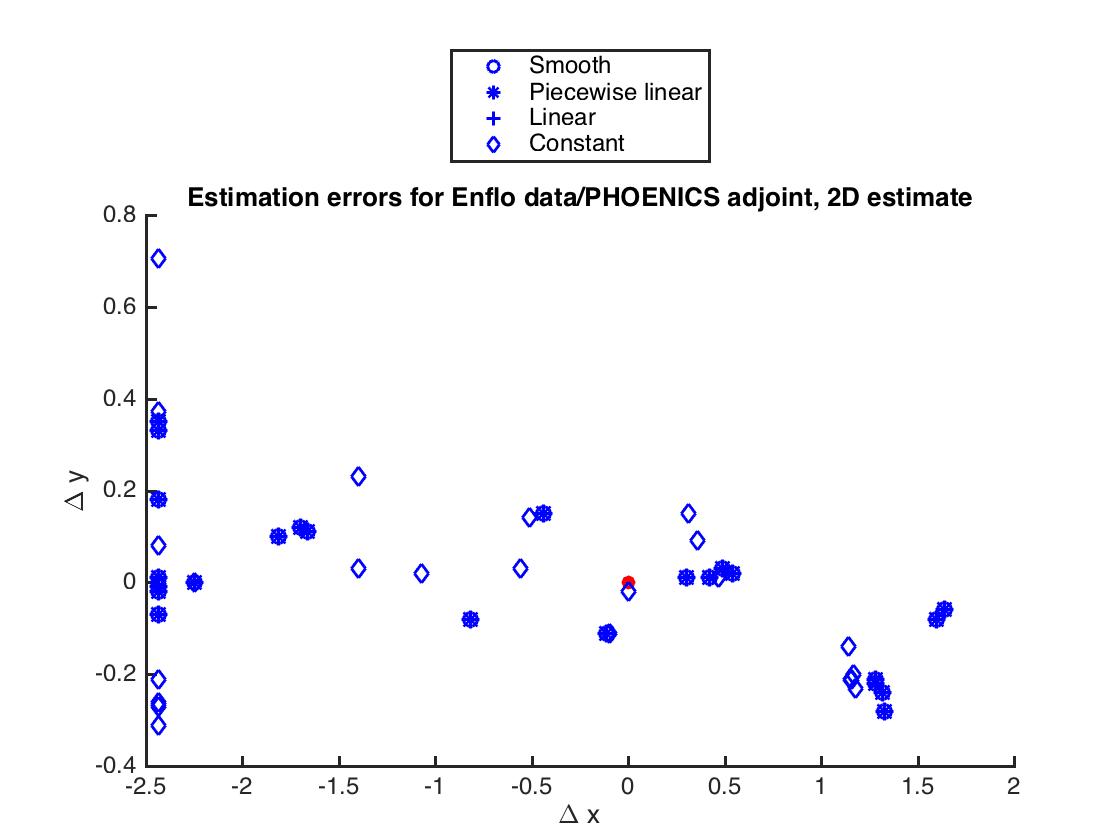}
\includegraphics[scale=0.15]{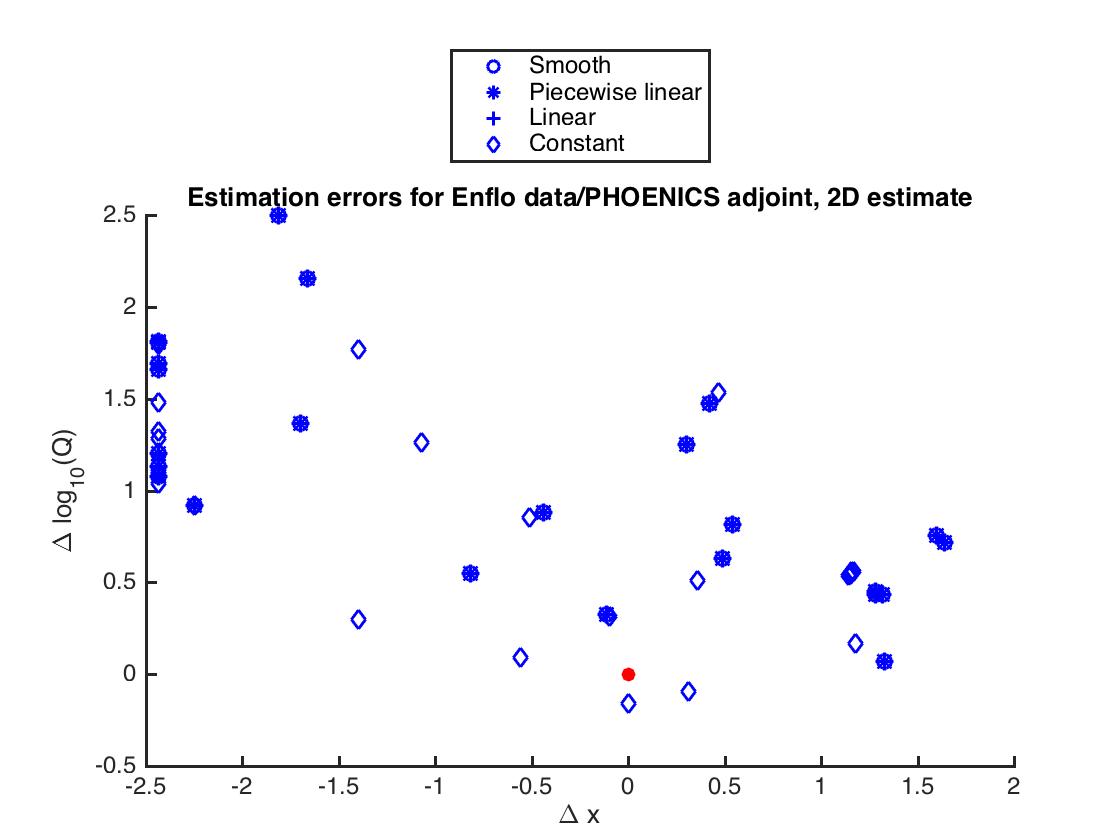}
\end{figure}

The average errors (averaged over all data sets) is presented in the
following table:

\begin{tabular}{lllll}
& Smooth & Piecewise linear & Linear & Constant \\ 
$x$ & $1.5$ & $1.5$ & $1.5$ & $1.41$ \\ 
$y$ & $0.117$ & $0.117$ & $0.117$ & $0.168$ \\ 
$\log _{10}\left( q\right) $ & $1.09$ & $1.09$ & $1.09$ & $0.853$%
\end{tabular}

\subsection{Results for $\protect\lambda =1$, $\protect\delta =10^{-6}$}

Increasing the penalty for large emission rates to $\lambda =1$, keeping $%
\delta =10^{-6}$, reduces the errors.

\subsubsection{3D estimates}

\begin{figure}[H]
\centering
\includegraphics[scale=0.15]{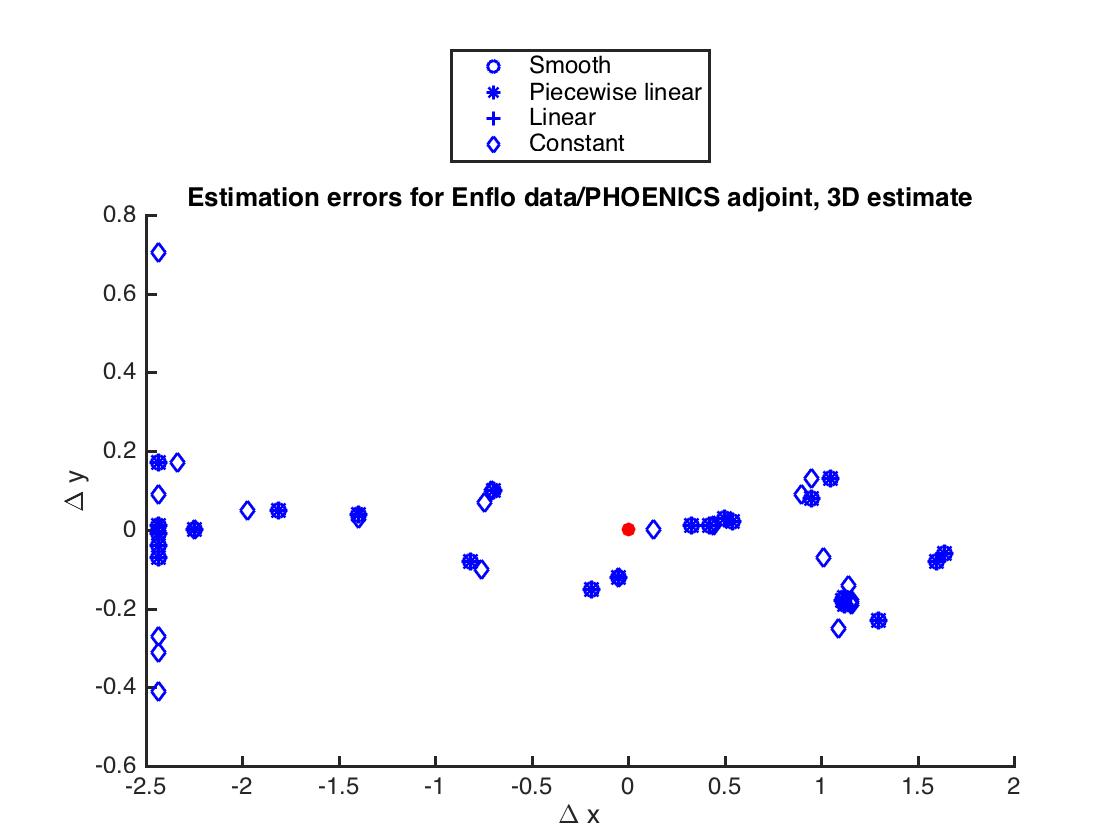}
\includegraphics[scale=0.15]{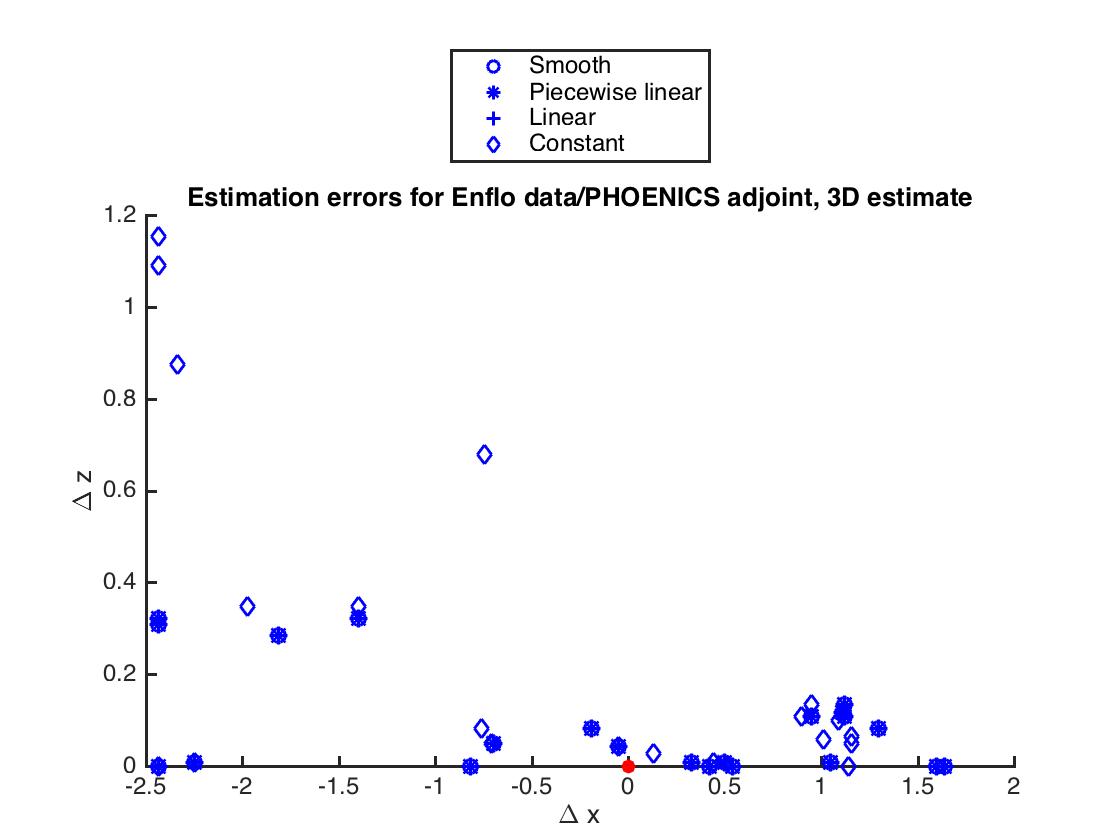}
\includegraphics[scale=0.15]{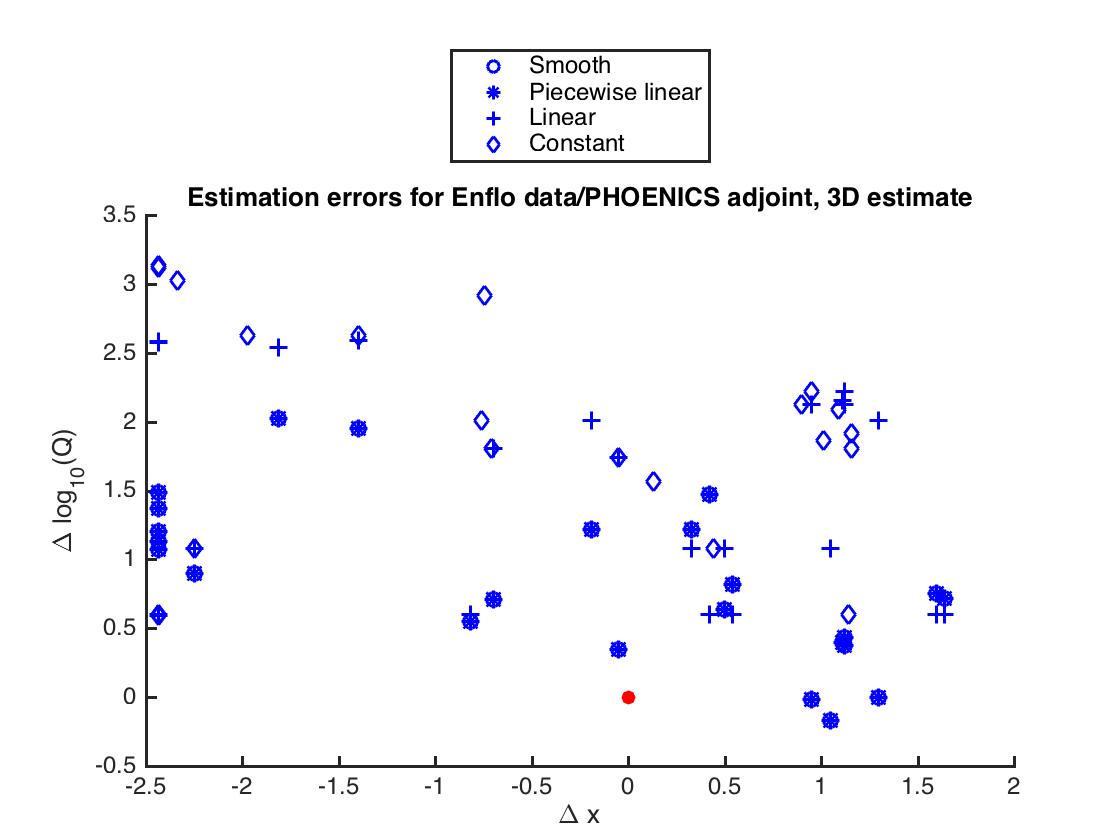}
\end{figure}

The average errors (averaged over all data sets) is presented in the
following table:

\begin{tabular}{lllll}
& Smooth & Piecewise linear & Linear & Constant \\ 
$x$ & $1.29$ & $1.29$ & $1.29$ & $1.42$ \\ 
$y$ & $0.0849$ & $0.0849$ & $0.0849$ & $0.22$ \\ 
$z$ & $0.0825$ & $0.0825$ & $0.0849$ & $0.22$ \\ 
$\log _{10}\left( q\right) $ & $0.873$ & $0.873$ & $1.48$ & $1.82$%
\end{tabular}

\subsubsection{2D estimates - assuming the source is on the ground}

\begin{figure}[H]
\centering
\includegraphics[scale=0.15]{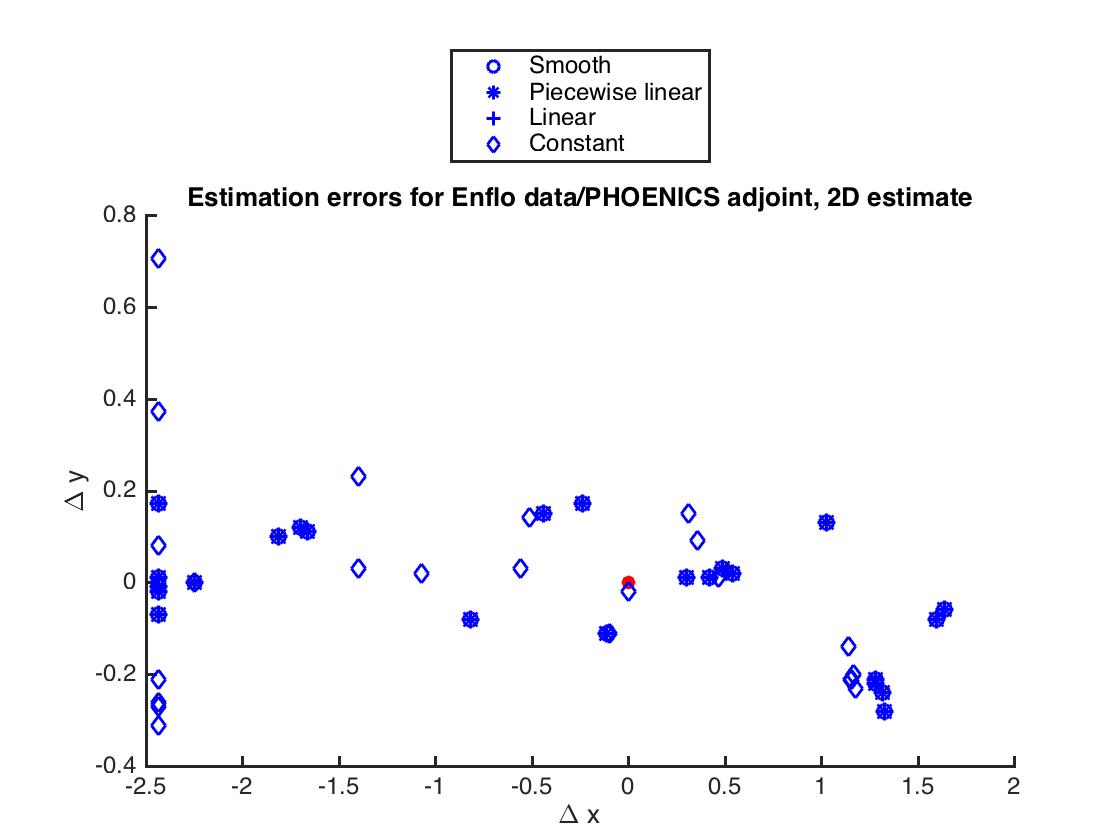}
\includegraphics[scale=0.15]{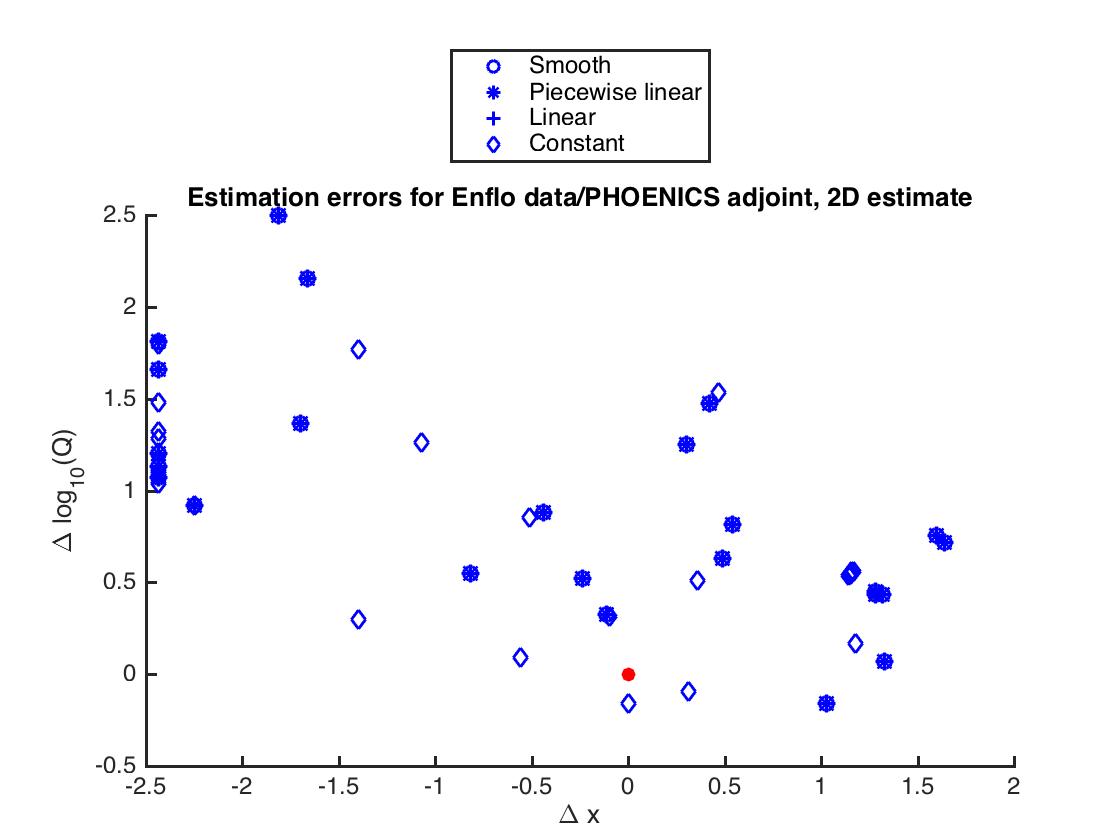}
\end{figure}

The average errors (averaged over all data sets) is presented in the
following table:

\begin{tabular}{lllll}
& Smooth & Piecewise linear & Linear & Constant \\ 
$x$ & $1.35$ & $1.35$ & $1.35$ & $1.41$ \\ 
$y$ & $0.101$ & $0.101$ & $0.101$ & $0.168$ \\ 
$\log _{10}\left( q\right) $ & $0.969$ & $0.969$ & $0.969$ & $0.853$%
\end{tabular}

\section{Discussion and Conclusions}

The bilevel optimisation method is well suited to linear inverse atmospheric
dispersion problems where a single source is emitting at a constant rate:
the problem decouples in a manner where the optimal source strength in each
grid point can be solved for first (follower problem) and then the optimal
location of the source can be determined (leader problem). When gauging how
good a candidate source is, namely how near its candidate sensor readings
are to the observed sensor readings - indeed how near its averaged candidate
sensor readings are to the averaged observed sensor readings, a distance
function is required. The distance function can be weighted, and since
averaged sensor readings are compared a natural choice of weight involves
the variance of the measurements. The simple choice is the let the weight
vary linearly with the model sensor reading. As the model sensor reading $%
\mu $ drops below the detection threshold however it is desired that the
variance of the real measurement $z$ becomes independent of the model sensor
reading $\mu $. The piecewise linear weight remedies this, but it comes at
the price of introducing a non-smoothness. We propose that this shortcoming
is addressed by introducing a smooth threshold weight keeping the benefit of
the piecewise linear weight (giving a variance with a strictly positive
lower bound). As shown in Figure \ref{fig:qualitative_difference} the
non-smoothness of the piecewise linear weight can yield a minimization
problem which is qualitatively different from smooth threshold distance
function and the linear one.

The method has been applied to experimental data from wind tunnel studies of
built up environments. Comparing average errors in the estimated parameters,
the smooth threshold distance function performs as well as the piecewise
linear one and usually better than the linear one. We also note that the
unweighted (constant) distance function is outperforming the others when it
comes to estimating the emission rate in the 2D case where it is a priori
assumed that the source is located on the ground. Increasing the penalty
coefficient $\lambda ,$penalising high emission rates, from $10^{-2}$ to $1$
decreases the average errors. It would be interesting to study how $\lambda $
should be chosen optimally.

\begin{acknowledgement}
This work was conducted within the European Defence Agency (EDA) project
B-1097-ESM4-GP \textquotedblleft Modelling the dispersion of toxic
industrial chemicals in urban environments\textquotedblright\ (MODITIC).
\end{acknowledgement}


\begin{thebibliography}{99}
\bibitem{Tikhonov1963} Tikhonov A. N. 1963 Translated in "Solution of
incorrectly formulated problems and the regularization method". \textit{%
Soviet Mathematics,} 4, 1035--1038.

\bibitem{Enting2002} Enting I G 2002 \textit{Inverse Problems in Atmospheric
Constituent Transport} Cambridge Univ Press, Cambridge

\bibitem{Marchuk1986} Marchuk G I 1986 Mathematical models in environmental
problems \textit{Studies in mathematics and its applications} \textbf{16}

\bibitem{GHL} Gudiksen P H, Harvey T F and Lange R 1989 Chernobyl source
term, atmospheric dispersion, and dose estimation \textit{Health Physics} 
\textbf{57} 5 697--706

\bibitem{StohlEtAl} Stohl A, Seibert P, Wotawa G, Arnold D, Burkhart J F,
Eckhardt S, Tapia C, Vargas A and Yasunari T J 2012 Xenon-133 and
caesium-137 releases into the atmosphere from the Fukushima Dai-ichi nuclear
power plant: determination of the source term, atmospheric dispersion, and
deposition \textit{Atmos. Chem. Phys.} \textbf{12} 2313--2343

\bibitem{RingbomEtAl} Ringbom A, Axelsson A, Aldener M, Auer M, Bowyer T W,
Fritioff T, Hoffman I, Khrustalev K, Nikkinen M, Popov Y, Ungar K and Wotawa
G 2014 Radioxenon detections in the VTBT international monitoring system
likely related to the announced nuclear test in North Korea on Febuary 12,
2013 \textit{Journal of Environmental Radioactivity} \textbf{128} 47--63

\bibitem{TheysEtAl} N. Theys, R. Campion, L. Clarisse, H. Brenot, J. van
Gent, B. Dils, S. Corradini, L. Merucci, P.-F. Coheur, M. Van Roozendael, D.
Hurtmans, C. Clerbaux, S. Tait, and F. Ferrucci 2013 Volcanic SO2 fluxes
derived from satellite data: a survey using OMI, GOME-2, IASI and MODIS 
\textit{Atmos. Chem. Phys.}, 13, 5945--5968

\bibitem{StohlEtAl2010} A. Stohl, A. J. Prata, S. Eckhardt, L. Clarisse, A.
Durant, S. Henne, N. I. Kristiansen, A. Minikin, U. Schumann, P. Seibert, K.
Stebel, H. E. Thomas, T. Thorsteinsson, K. T\o rseth, and B. Weinzierl \
2011 Determination of time- and height-resolved volcanic ash emissions and
their use for quantitative ash dispersion modeling: the 2010 Eyjafjallaj\"{o}%
kull eruption, \textit{Atmos. Chem. Phys.}, 11, 4333--4351

\bibitem{GrahnEtAl} Grahn H., von Schoenberg P., Br\"{a}nnstr\"{o}m N. 2015
What's that Smell? Hydrogen sulphide transport from Bardarbunga to
Scandinavia \textit{Journal of Volcanology and Geothermal Research}, 303,
187--192

\bibitem{Stuart2010} Stuart A M 2010 Inverse problems: a Bayesian
perspective \textit{Acta Numerica} \textbf{19}

\bibitem{KYL2007} Keats A, Yee E and Lien F-S 2007 Bayesian inference for
source determination with applications to a complex urban environment 
\textit{Atmospheric Environment} \textbf{41} 465--479

\bibitem{Yee2007} Yee E 2007 \textit{Bayesian Inversion of Concentration
Data for an Unknown Number of Contaminant Sources} Technical Report DRDC
Suffield TR 2007-085

\bibitem{YF2010} Yee E and Flesch T K 2010 Inference of emission rates from
multiple sources using Bayesian probability theory \textit{Journal of
Environmental Monitoring} \textbf{12} 622--634

\bibitem{Yee2012} Yee E 2012 Probability Theory as Logic: Data Assimilation
for Multiple Source Reconstruction \textit{Pure and Applied Geophysics} 
\textbf{169} 499--517

\bibitem{Yee2012B} Yee E 2012 Inverse Dispersion for an Unknown Number of
Sources: Model Selection and Uncertainty Analysis \textit{ISRN Applied
Mathematics} \textbf{2012}

\bibitem{BP2015} Br\"{a}nnstr\"{o}m N and Persson L \AA\ 2015 A measure
theoretic approach to linear inverse atmospheric dispersion problems \textit{%
Inverse Problems, }\textbf{31}, 2.

\bibitem{RL1998} Robertson L and Langner J 1998 Source function estimate by
means of variational data assimilation applied to the ETEX-I tracer
experiment \textit{Atmospheric Environment} \textbf{32} 24 4219--4225

\bibitem{THG2007} Thomson L C, Hirst B, Gibson G, Gillespie S, Jonathan P,
Skeldon K D and Padget M J 2007 An improved algorithm for locating a gas
source using inverse methods \textit{Atmospheric Environment} \textbf{41} 6
1128--1134

\bibitem{AYH2007} Allen C T, Young G S and Haupt S E 2007 Improving
pollutant source characterization by better estimating wind direction with a
genetic algorithm \textit{Atmospheric Environment} \textbf{41} 11 2283--2289

\bibitem{ISS2012} Issartel J P, Sharan M and Singh S K 2012 Indentification
of a point source by use of optimal weighted least squares \textit{Pure and
Applied Geophysics} \textbf{169} 467--482

\bibitem{SSI2012} Sharan M, Singh S K and Issartel J P 2012 Least square
data assimilation for identification of the point source emissions \textit{%
Pure and Applied Geophysics} \textbf{169} 483--497

\bibitem{Bocquet2005} Bocquet M 2005 Reconstruction of an atmospheric tracer
source using the principle of maximum entropy I: theory \textit{Quarterly
Journal of the Royal Meteorological Society} \textbf{131} 610B 2191--2208

\bibitem{Bard1998} Jonathan F. Bard:\ Practical Bilevel Optimization -
Algorithms and Applications. Springer, 1998.

\bibitem{MilgromSegal2002} Paul Milgrom and Ilya Segal:\ Envelope theorems
for arbitrary choice sets. Econometrica, Vol. 70, No. 2, March 2002, p.
583-601.

\bibitem{NocedalWright2006} Nocedal J and Wright S J 2006 \textit{Numerical
Optimization}, Second Edition, Springer

\bibitem{RobinsEtAl2016a} Robins, A., M. Carpentieri, P. Hayden, J. Batten,
J. Benson and A. Nunn, 2016: MODITIC WIND TUNNEL EXPERIMENTS, FFI Report (in
preparation). see also Harmo17 extended abstract with the same title.

\bibitem{RobinsEtAl2016b} Robins, A., M. Carpentieri, P. Hayden, J. Batten,
J. Benson and A. Nunn, 2016: MODITIC WIND TUNNEL EXPERIMENTS, Conference
proceedings of the 17th International Conference on Harmonisation within
Atmospheric Dispersion Modelling for Regulatory Purposes (HARMO17).

\bibitem{BurkhartBurman} S. Burkhart and J. Burman, 2016, MODITIC WIND
TUNNEL EXPERIMENTS NEUTRAL AND HEAVY GAS SIMULATION USING RANS, Conference
proceedings of the 17th International Conference on Harmonisation within
Atmospheric Dispersion Modelling for Regulatory Purposes (HARMO17)

\bibitem{WP7000} N. Br\"{a}nnstr\"{o}m, S. Burkhart, J. Burman, X. Busch,
J.-P. Issartel, L. \AA . Persson, 2016, MODITIC WP7000 Backtracking - Linear
inverse modelling of neutral gases in urban environments, FOI Report (in
preparation).
\end{thebibliography}
\end{document}